\documentclass[11pt,fleqn]{article}
\usepackage{amssymb,latexsym,amsmath,amsfonts,graphicx}
\usepackage{pictex,color,comment}

    \advance\textheight by \topskip

    \numberwithin{equation}{section}

\def\beq{\begin{equation}}
\def\eeq{\end{equation}}
\def\nn{\nonumber}
\def\bit{\begin{itemize}}
\def\eit{\end{itemize}}

\makeatletter
\def\eqalign#1{\null\vcenter{\def\\{\cr}\openup\jot\m@th
  \ialign{\strut$\displaystyle{##}$\hfil&$\displaystyle{{}##}$\hfil
      \crcr#1\crcr}}\,}
\makeatother

\newcommand{\bea}{\begin{align}}
\newcommand{\eea}{\end{align}}

\newcommand{\be}{\begin{equation}}
\newcommand{\ee}{\end{equation}}

    \def\e{{\epsilon}}

    \def\Re{{\rm Re \,}}
    \def\Im{{\rm Im \,}}
    \def\Ai{{\rm Ai \,}}

    \def\bigO{{\cal O}}

    \def\Res{{\rm Res}}
    \def\m{{(m)}}
    \def\P2n{{\rm P}_{{\rm II}}^{(n)}}

    \newtheorem{theorem}{Theorem}[section]
    \newtheorem{lemma}[theorem]{Lemma}
    \newtheorem{corollary}[theorem]{Corollary}
    \newtheorem{proposition}[theorem]{Proposition}
    
    \newtheorem{Definition}[theorem]{Definition}
    
    \newtheorem{Remark}[theorem]{Remark}
    \newenvironment{remark}{\begin{Remark}\rm}{\end{Remark}}
    \newtheorem{Example}[theorem]{Example}
    \newenvironment{example}{\begin{Example}\rm}{\end{Example}}
    \newtheorem{Assumptions}[theorem]{Assumptions}

    \newenvironment{proof}%
    {\rm \trivlist \item[\hskip \labelsep{\bf Proof. }]}%
    {\hspace*{\fill}$\Box$\endtrivlist}
    {\rm \trivlist \item[\hskip \labelsep{\bf Proof}]}%
    {\hspace*{\fill}$\Box$\endtrivlist}

    \DeclareMathOperator*{\Tr}{Tr}
    
\begin{document}
\title{Random matrix ensembles with singularities and a hierarchy of Painlev\'{e} III equations}
\author{Max R. Atkin, Tom Claeys and Francesco Mezzadri}
\maketitle

\begin{abstract}
We study unitary invariant random matrix ensembles with singular potentials. We obtain asymptotics for the partition functions associated to the Laguerre and Gaussian Unitary Ensembles perturbed with a pole of order $k$ at the origin, in the double scaling limit where the size of the matrices grows, and at the same time the strength of the pole decreases at an appropriate speed. In addition, we obtain double scaling asymptotics of the correlation kernel for a general class of ensembles of positive-definite Hermitian matrices perturbed with a pole. Our results are described in terms of a hierarchy of higher order analogues to the Painlev\'e III equation, which reduces to the Painlev\'e III equation itself when the pole is simple.
\end{abstract}

\section{Introduction and statement of results}
We study unitary invariant random matrix ensembles on the space $\mathcal H_n^+$ of $n\times n$ positive-definite Hermitian matrices defined by the probability measure
\beq\label{pUE}
\frac{1}{C_n} (\mathrm{det} M)^\alpha \exp \left[-n\Tr V_k(M)\right] dM, \qquad \alpha > -1,
\eeq
with
\beq
dM=\prod_{j=1}^n d\Re M_{jj}\prod_{1\leq i<j\leq n}d\Re M_{ij} d\Im M_{ij},
\eeq
and
\beq C_n=\int_{\mathcal H_n^+}(\mathrm{det} M)^\alpha \exp \left[-n\Tr V_k(M)\right] dM,
\eeq
in the case where the potential $V_k(x)$ has a pole of order $k$, i.e.\ 
\beq
V_k(x) = V(x) + \left(\frac{t}{x}\right)^k. 
\eeq
The regular part $V$ of the potential is a real analytic function on $[0,+\infty)$ subject to some constraints which we will detail later. In particular, we will require $V$ to be such that, for $t=0$, the limiting mean density of the eigenvalues as $n\to\infty$ is supported on an interval of the form $[0,b]$, with $b>0$.

It is well known that the eigenvalues of a random matrix drawn from the ensemble (\ref{pUE}) form a determinantal point process. 
The joint probability distribution of the eigenvalues is given by
\beq\label{jpdf}
\frac{1}{Z_{n,k}}\Delta(x)^2\prod_{j=1}^n x_j^\alpha e^{-nV_k(x_j)}dx_j,\qquad \Delta(x)=\prod_{1\leq i<j\leq n} (x_j-x_i),
\eeq
with the partition function $Z_{n,k}=Z_{n,k}(V)$ given by
\beq\label{partition function}
Z_{n,k}=\int_{[0,+\infty)^n}\Delta(x)^2 \prod_{j=1}^n x_j^\alpha e^{-nV_k(x_j)}dx_j.
\eeq

For $t>0$ small, the model (\ref{pUE}) can be seen as a singular perturbation of the unitary invariant ensemble corresponding to $t=0$.
For $t>0$, the eigenvalues are pushed away from $0$ because of the pole in the potential, and for large $n$, the probability of finding eigenvalues close to $0$ is small if $t>0$ is independent of $n$. However, if $t\to 0$ together with $n\to\infty$, this repulsion becomes weaker and one expects a transition between a regime where eigenvalues are likely to be found in the vicinity of the origin, and one where eigenvalues are unlikely to be found near the origin.

The effect of singular perturbations of unitary  invariant ensembles has been of recent interest \cite{CI,BMM,XDZ, XDZ2}.  In \cite{CI}, the singularly perturbed Laguerre Unitary Ensemble (pLUE) was studied, given by the measure
\beq
\label{pLUE}
\frac{1}{C_n} (\mathrm{det} M)^\alpha \exp \left[-n\Tr \left(M + \frac{t}{M}\right) \right] dM,
\eeq
on the space $\mathcal H_n^+$, where $C_n$ is a normalisation constant. A relation between this model and the Painlev\'{e} III (henceforth PIII) equation was established for fixed $n$ in \cite{CI}. In subsequent work \cite{BMM}, a singular perturbation of the Gaussian Unitary Ensemble, which we will refer to as pGUE, was studied, defined by the measure
\beq
\frac{1}{\widehat C_n}  \exp \left[-n\Tr \left(\frac{1}{2}M^2 + \frac{t}{2 M^2}\right) \right]  dM,
\eeq
on the set of $n\times n$ Hermitian matrices $\mathcal H_n$. 

 In \cite{BMM} the double scaling limit where $t \rightarrow 0$ as $n \rightarrow \infty$ of the partition function was analysed using Riemann-Hilbert (RH) techniques. A connection to PIII was also found by relating the orthogonal polynomials associated to the pGUE measure to those of the pLUE measure. 
In \cite{XDZ}, the double scaling limit for the eigenvalue correlation kernel in the pLUE was studied. A limiting kernel was found, defined in terms of a model RH problem associated to a special solution of the PIII equation. This limiting kernel degenerates to the Airy kernel if $t$ approaches $0$ at a slow rate, and to the Bessel hard edge kernel if $t\to 0$ at a fast rate.
In \cite{XDZ2}, asymptotics for the partition function in the pLUE, again in terms of a PIII transcendent, were obtained.

In the present work, we will obtain asymptotics for the eigenvalue correlation kernel for a fairly general class of potentials $V$, perturbed with a pole of order $k\in\mathbb N$, in a double scaling limit where the strength of the perturbation, $t$, goes to zero at an appropriate speed as the
size of the matrix, $n$, is taken to infinity. The double scaling limit will be tuned in such a way that the macroscopic behaviour of the eigenvalues in the large $n$ limit is the same as in the non-singular case where $t=0$, but such that the microscopic behaviour of the eigenvalues near $0$ is affected by the singularity of the potential.
In addition, we will obtain double scaling asymptotics for the partition functions associated to the LUE perturbed with a pole of order $k$, and to the GUE perturbed with a pole of order $2k$.
Our results will be described in terms of special solutions to a family of systems of ODEs, indexed by $k\in\mathbb N$, which can be seen as a hierarchy of Painlev\'e III equations.

\subsection{Motivations}
There are a number of motivations for considering the model (\ref{pUE}). Firstly, it was observed that critical one-matrix models in which the limiting mean density of eigenvalues vanishes in the bulk of the spectrum or at a higher than generic order at the edge of the spectrum, are in one-to-one correspondence with (super-)Liouville field theories. In this context, it is natural to look for models with higher order critical points in which multiple zeros appear in the bulk or at the edge of the spectral density, as it is known that such higher order critical models correspond to coupling a minimal conformal field theory to (super)-Liouville theory. A review of these facts in the non-supersymmetric case can be found in \cite{Di Francesco:1993nw}. Supersymmetric versions followed later in \cite{Klebanov:2003wg} and are nicely reviewed in the appendix of \cite{Seiberg:2004ei}. One-matrix models in which the potential has poles exhibit a new type of critical behaviour. Although at present no conformal field theory analogue of such models is known, it is again natural to study higher order critical points in this context.

A second physical motivation arises in the field of quantum transport and electrical characteristics of chaotic cavities. Here the quantity of interest is the Wigner-Smith time-delay matrix $Q$, the eigenvalues of which, $\tau_j$, are known as the ``partial delay times''. In systems in which the dynamics is chaotic, a RMT approach has been quite successful and one of the central results of this approach is the joint probability density for the inverse delay times, $\gamma_j = \tau_j^{-1}$, first obtained in \cite{BFB1,BFB2}. It takes the form
\beq
P(\gamma_1, \ldots, \gamma_n) =\frac{1}{C_{n,\beta}} \left |\Delta(\gamma)\right |^\beta \prod_{j=1}^n \gamma_j^{\beta n/2} e^{-\frac{\beta}{2} \gamma_j},
\eeq
where $\beta$ depends on the symmetries of the system, with $\beta=2$ a common case. Since many observables may be expressed in terms of $Q$, the problem of computing expectation values with respect to the above measure is relevant. In particular the observable
\beq
\tau_W = \frac{1}{n}\Tr Q = \frac{1}{n} \sum^n_{i = 1} \tau_i,
\eeq
known as the Wigner time-delay has been considered in the recent work \cite{MT}. The partition function
for our model \eqref{pUE} coincides with the moment generating function for the probability density of $\tau_W$ in the case $k=1$ and $\alpha = \beta n/2$. Although we will not scale $\alpha$ with $n$ in the present paper, the moment generating function in this model is the partition function for a matrix model which shares the feature of a singular potential with our model (\ref{pUE}).

Furthermore, in a very recent paper \cite{GT}, an observable $R_q$ known as the ``charge relaxation resistance'' defined as
\beq
R_q \propto \frac{\Tr Q^2}{(\Tr Q)^2}=\frac{\sum_{j=1}^n\tau_j^2}{(\sum_{j=1}^n \tau_j)^2}
\eeq
was considered. To compute $R_q$, the approach taken in \cite{GT} was to compute the distributions of $\Tr Q^2$ and $\Tr Q$ separately. The associated moment generating function is the partition function for a perturbed pLUE model with a singularity in the potential of order $k=2$, thereby demonstrating the physical relevance of the model studied here for $k>1$.

A third motivating model appears in the field of spin-glasses \cite{AVV}. Here a model corresponding to the GOE perturbed by a pole of order $k$ was analysed for its relation to the distribution of the spin glass susceptibility in the Sherrington-Kirkpatrick (SK) mean-field model. The GUE version of the partition function in such a model relates directly to the partition function in the pLUE, as we will see later on in this paper. One-matrix models in which the potential has poles have also appeared in the context of replica field theories~\cite{OsiKan}.

Finally, our last motivation for this work is to seek a natural candidate for a PIII hierarchy. The notion of a PIII hierarchy has appeared little in the literature; one of the few mentions being \cite{Sakka}. This work was partly motivated by the desire to understand whether the hierarchy proposed in  \cite{Sakka} would appear when the order of the pole was increased and if not, what form the alternative hierarchy takes. It appears that the hierarchy of equations which we will obtain is different from the one in \cite{Sakka}.


\subsection{Statement of results}

Our main results are the following.
\begin{enumerate}
\item We define a hierarchy of higher order PIII equations and prove the existence of special pole-free solutions to it. 
\item We obtain double scaling asymptotics for the partition  function in the LUE perturbed with a pole of order $k\geq 1$, in terms of a higher order PIII transcendent. This generalizes the result from \cite{XDZ} for $k=1$.
\item We obtain double scaling asymptotics for the partition function in the GUE perturbed with a pole of order $2k$, $k\geq 1$. They are again given in terms of higher order PIII transcendents. In the case $k=1$, such asymptotics were already obtained in \cite{BMM}, but written in a different form.
\item In the double scaling limit, we prove, for $k\geq 1$ and for general $V$, that the eigenvalue correlation kernel near $0$ of the model (\ref{pUE}) tends to a limiting kernel built out of a model RH problem associated to the PIII hierarchy.
This extends the result from \cite{XDZ} for $V(x)=x$ and $k=1$.
\end{enumerate}
We now state our results in more detail.

\subsubsection*{A Painlev\'{e} III hierarchy}

Given $k \in \mathbb{N}$, consider the system of $k+1$ ODEs indexed by $p=0,\ldots, k$,
\beq\label{P3def}
\sum_{q=0}^{p} \left(\ell_{k-p+q+1}\ell_{k-q}-(\ell_{k-p+q} \ell_{k-q})'' + 3\ell_{k-p+q}' \ell_{k-q}' - 4u \ell_{k-p+q} \ell_{k-q}\right) = \tau_p,
\eeq
for $k+1$ unknown functions $u=u(s)$ and $\ell_1=\ell_1(s), \ldots, \ell_k=\ell_k(s)$, with \beq \label{initial conditions}\ell_{k+1}(s) = 0,\qquad  \ell_0(s) = \frac{s}{2}.
\eeq The $\tau_p$'s are real constants that play the role of times. The $p=0$ equation always results in
\beq
u = -\frac{1}{4 \ell_k^2}\left((\ell_k^2)'' - 3 (\ell_k')^2 + \tau_0\right).
\eeq
Substituting this expression for $u$ in the other equations, we are left with $k$ equations for $k$ unknowns $\ell_1, \ldots, \ell_k$. We refer to this system of equations as the $k$-th member of the Painlev\'{e} III hierarchy. Eliminating $\ell_p$ for $2 \leq p \leq k$ leads to an ODE for $\ell_1$ of order $2k$.

\begin{example}
For $k=1$ we have the equation
\beq
\ell_1''(s) = \frac{\ell_1'(s){}^2}{\ell_1(s)}-\frac{\ell_1'(s)}{s}-\frac{\ell_1(s){}^2}{s}-\frac{\tau _0}{\ell_1(s)}+\frac{\tau _1}{s},
\eeq
which we identify as a special case of the Painlev\'{e} III equation, see \cite{FIKN}.
\end{example}
\begin{example}
If $k=2$ we have a system of two ODEs;
\beq
\frac{\tau _1}{2 \ell_1(s) \ell_2(s)}-\frac{\tau _0}{\ell_2(s){}^2}+\frac{\ell_2'(s){}^2}{\ell_2(s){}^2}-\frac{\ell_1'(s) \ell_2'(s)}{\ell_1(s) \ell_2(s)}+\frac{\ell_1''(s)}{\ell_1(s)}-\frac{\ell_2''(s)}{\ell_2(s)}-\frac{\ell_2(s)}{2\ell_1(s)} = 0, 
\eeq
and
\begin{align}
\frac{\ell_1(s){}^2 \ell_2'(s){}^2}{\ell_2(s){}^2}&-\ell_1'(s){}^2+\frac{s \ell_2'(s){}^2}{\ell_2(s)}-\ell_2'(s)-\frac{\tau _0 \ell_1(s){}^2}{\ell_2(s){}^2}-\frac{s \tau _0}{\ell_2(s)}-\tau _2 \nn\\
&=\frac{2\ell_1(s){}^2 \ell_2''(s)}{\ell_2(s)}-2 \ell_1(s) \ell_1''(s)+s \ell_2''(s)+2 \ell_2(s) \ell_1(s).
\end{align}
One can eliminate $\ell_2$ in order to obtain a single equation of order four for $\ell_1$.
\end{example}

We can construct, for any $k=1,2,\ldots $, a special set of solutions $\ell_1, \ldots, \ell_k$ to the $k$-th member of the PIII hierarchy in terms of a model RH problem. The function $\ell_1(s)$ will be of particular importance to us. The model RH problem consists of finding a function $\Psi=\Psi(z;s)$ satisfying the following properties.


\subsubsection*{RH problem for $\Psi$}
\begin{itemize}
\item[(a)] $\Psi:\mathbb C\setminus\Sigma\to\mathbb C^{2\times 2}$ analytic, with $\Sigma=\cup^{3}_{i=1} \Sigma_i\cup\{0\}$ as illustrated in Figure \ref{modelcontour}. The half-lines $\Sigma_1,\Sigma_3$ can be chosen freely in the upper and lower half plane.
\item[(b)] $\Psi$ has continuous boundary values $\Psi_\pm(z)$ as $z\in\Sigma\setminus\{0\}$ is approached from the left ($+$) or right ($-$) side of $\Sigma\setminus\{0\}$, and they are related by
\begin{align}
&\Psi_+(z)=\Psi_-(z)\begin{pmatrix}1&0\\-e^{\pi i\alpha}&1\end{pmatrix}, &z\in\Sigma_1,\\
&\Psi_+(z)=\Psi_-(z)\begin{pmatrix}0&-1\\1&0\end{pmatrix}, &z\in \Sigma_2,\\
&\Psi_+(z)=\Psi_-(z)\begin{pmatrix}1&0\\-e^{-\pi i\alpha}&1\end{pmatrix}, &z\in\Sigma_3.
\end{align}
\item[(c)] As $z\to\infty$, there exist functions $p(s), q(s), r(s)$ such that $\Psi$ has the asymptotic behaviour
\begin{equation}\label{Psic}
\Psi(z)=\left(I+\frac{1}{z}\begin{pmatrix}q(s)&-i r(s)\\i p(s)&-q(s)\end{pmatrix}+\bigO(z^{-2})\right)z^{-\frac{1}{4}\sigma_3}Ne^{z^{1/2}\sigma_3},
\end{equation}
where $N=\frac{1}{\sqrt{2}}(I+i\sigma_1)$, with $\sigma_1=\begin{pmatrix}0&1\\1&0\end{pmatrix}$, and where the principal branches of $z^{1/2}$ and $z^{-1/4}$ are taken, analytic off $(-\infty,0]$ and positive for $z>0$. The third Pauli matrix $\sigma_3$ is given by $\sigma_3=\begin{pmatrix}1&0\\0&-1\end{pmatrix}$.
\item[(d)] As $z\to 0$, there exists a matrix $\Psi_0(s)$, independent of $z$, such that $\Psi$ has the asymptotic behaviour
\begin{equation}
\label{Psi0}
\Psi(z)=\Psi_0(s)(I+\bigO(z))e^{-\left(-\frac{s}{z}\right)^k\sigma_3}z^{\frac{\alpha}{2} \sigma_3}H_j,
\end{equation}
for $z \in \Omega_j$, where $H_1, H_2, H_3$ are given by
\begin{align}
&\label{H1}H_1=I\\
&\label{H2}H_2=\begin{pmatrix}1&0\\-e^{\pi i\alpha}&1\end{pmatrix},\\
&\label{H3}H_3=\begin{pmatrix}1&0\\e^{-\pi i\alpha}&1\end{pmatrix}.
\end{align}
\end{itemize}
\begin{figure}[t]
\centering 
\includegraphics[scale=0.6]{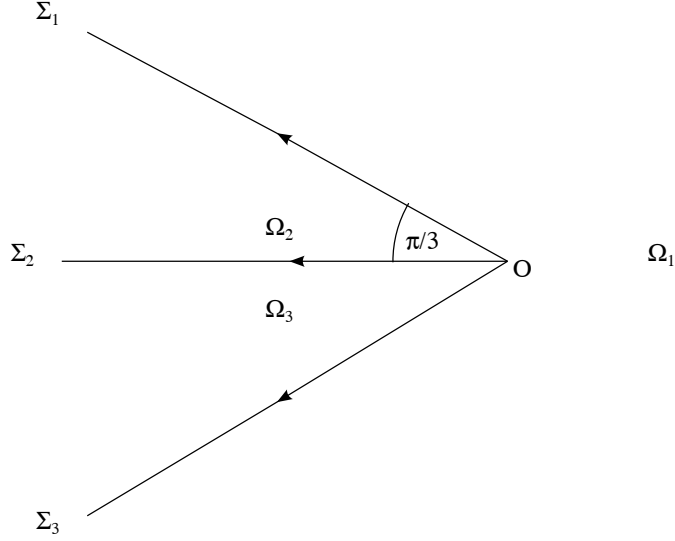}
\caption{The jump contour $\Gamma$ for the model RH problem for $\Psi$. Contours are labelled by $\Sigma$ and sectors by $\Omega$.}
\label{modelcontour}
\end{figure}
\begin{remark}
In the case $k=1$, the RH problem for $\Psi$ coincides, up to the orientation of the contours, with the model RH problem introduced in \cite{XDZ}.
In the case $k=0$, it is a RH problem which can be solved using Bessel functions \cite{Vanlessen2}.
\end{remark}
\begin{remark}
It is important to note that the function
\beq
\Psi(z,s) H_j^{-1} z^{-\frac{\alpha}{2} \sigma_3}e^{\left(-\frac{s}{z}\right)^k\sigma_3}
\eeq
is not analytic at $z=0$ and in fact has a jump across $\Gamma_2$. Indeed, if the asymptotic behaviour in condition (d) of the above RH problem were modified to
\begin{equation}
\Psi(z,s)=\Psi_0(z, s)e^{-\left(-\frac{s}{z}\right)^k\sigma_3}z^{\frac{\alpha}{2} \sigma_3}H_j,
\end{equation}
with $\Psi_0(z,s)$ analytic at $z = 0$, then the resulting RH problem would have no solution. A more detailed description of the analytic structure of $\Psi$ near the origin can be given as follows. Define 
\begin{equation}\label{def f origin}
f_2(z,s)=\frac{e^{-z}}{2\pi i}\int_0^{-\infty}|u|^\alpha e^u e^{-2\left(-\frac{s}{u}\right)^k}\frac{du}{u-z},
\end{equation}
and let $\widehat\Psi_0$ be defined by
\begin{equation}\label{Psi0better}
\Psi(z,s)=\widehat\Psi_0(z,s)\begin{pmatrix}1&f_2(z,s)\\0&1
\end{pmatrix}e^{-\left(-\frac{s}{z}\right)^k\sigma_3}z^{\frac{\alpha}{2}\sigma_3} H_j.
\end{equation}
Then $\widehat\Psi_0$ is an analytic function near $z=0$, and $\Psi_0$ defined by (\ref{Psi0}) takes the form
\begin{equation}
\Psi_0(s)=\widehat\Psi_0(0,s)\begin{pmatrix}1&f_2(0,s)\\0&1\end{pmatrix}.
\end{equation}
\end{remark}
\begin{remark}
If the model RH problem has a solution, it follows from standard techniques that the solution is unique. Existence of a solution is a much more subtle issue. We will show that the model RH problem is solvable for $k\in\mathbb N$ and $s>0$.
\end{remark}
\begin{theorem}\label{P3theorem}
Let $\alpha>-1$, and let $\Psi(z;s)$ be the unique solution of the model RH problem for $s>0$. Then, the limit
\beq
\label{qPhi}
y_\alpha(s) = -2 i \frac{d}{ds}\left[\lim_{z \rightarrow \infty} z \Psi(z,s^2) e^{-z^{1/2}\sigma_3} N^{-1}z^{\frac{1}{4}\sigma_3}\right]_{12}=-2\frac{d}{ds}\left(r(s^2)\right)
\eeq
exists and it is a solution of the equation for $\ell_1$ in the $k$-th member of the Painlev\'{e} III hierarchy, with the parameters $\tau_0,\ldots, \tau_k$ given by
\beq \label{tau}
\tau_p =  
\begin{cases}
4^{2k+1} k^2,&\mbox{ for $p = 0$,}\\
-(-4)^{k+1} \alpha k,&\mbox{ for $p = k$,}\\
0,&\mbox{ for $0 < p < k$.}\end{cases}
\eeq
Moreover, $y(s)=y_\alpha(s)$ has the following asymptotics as $s\to +\infty$ and as $s\to 0$,
\begin{align}
&\label{as y infty}y(s)=-\frac{8k}{2k+1}\left(\beta_{k-2} - \frac{3}{2}z_0\right)s^{\frac{2k-1}{2k+1}} + \bigO(1),& s\to +\infty,\\
&\label{as y 0}y(s)=\bigO(s^{2k-1})+\bigO(s^{2\alpha+1}),& s\to 0,
\quad  s>0,
\end{align}
where we have used the constants
\begin{align}
&z_0 = - \left(\frac{2^{k-1}(k-1)!}{(2k-1)!!}\right)^{-\frac{2}{2k+1}}, \\
&\beta_{j} = -(-z_0)^{-\frac{3}{2}-j}\frac{(2j+1)!!}{2^{j} j!},
\end{align}
with $(2j+1)!!=(2j+1)(2j-1)\ldots 3.1$ the double factorial.
\end{theorem}
\begin{remark}
\label{lnremark}
In fact we will prove a more general result in the sense that all $\ell_j$'s may be extracted from the model problem for $\Psi$. To this end we write the asymptotic expansion of $\Psi$ as $z \to \infty$ as
\begin{equation}
\Psi(z)=\left(I+\sum^\infty_{j=1}C_j z^{-j}\right)z^{-\frac{1}{4}\sigma_3}Ne^{z^{1/2}\sigma_3}, \qquad C_j = \begin{pmatrix}q_j(s)&-i r_j(s)\\i p_j(s)&-q_j(s)\end{pmatrix}.
\end{equation}
Define the formal power series
\beq
m_2(z,s) = \sum^\infty_{j=1} \frac{r_j(s^2)}{s^{2j-1}z^{j}},
\eeq
and define in addition formal power series $m_1,m_3,m_4$ in $z^{-1}$, $m_i(z,s) = \sum^\infty_{j=1} m_{i,j}(s) z^{-j}$. 
The quantities $m_3$ and $m_4$ are defined in terms of $m_1$ and $m_2$ as
\begin{align}
&m_3(z,s) := -\partial_s m_1-\frac{r\left(s^2\right) (m_1+1)}{s}-\frac{1}{2} \partial^2_s m_2\nonumber \\&\label{k3rule}\hspace{5cm} +m_2 \left(3 r'\left(s^2\right)-\frac{r\left(s^2\right){}^2}{2
s^2}-\frac{3 r\left(s^2\right)}{2 s^2}+z\right), \\
\label{k4rule}
&m_4(z,s) := \frac{1}{2}\left(\frac{r(s^2)}{s} m_2 - \partial_s m_2 \right),
\end{align} and $m_1$ can be found recursively from the relation
\beq
\label{kdef}
m_{1,j} = \frac{1}{2} (m_4^2 + m_2 m_3)_j - \frac{1}{2} \sum^{j-1}_{i=0} m_{1,j-i}m_{1,i},
\eeq
where $(m_4^2 + m_2 m_3)_j$ denotes the coefficient of $z^{-j}$ in the formal power series of $m_4^2 + m_2 m_3$. Note that upon substituting in the expressions for $m_3$ and $m_4$ in terms of $m_1$ and $m_2$, the right hand side of \eqref{kdef} only contains $m_{1,i}$ with $i < j$ and therefore gives a well defined recursive relation for $m_{1,j}$.  

We now introduce the matrix
\beq
K(z,s) := \begin{pmatrix}m_1+m_4+1&m_2\\m_3+ s^{-1}r(s^2)(m_1+m_4+1)&m_1-m_4+s^{-1}r(s^2)m_2+1 \end{pmatrix},
\eeq
where $r=r_1$. We then obtain expressions for all $\ell_j$ in terms of $r_j$ via the relation
\begin{align}
4\sum^j_{n=0} &\ell_{j-n}(s) (4z)^n = \\
&\left((4z)^{j+1}\Tr\left[(\partial_z K) K^{-1} +\frac{s}{2} K \sigma_- K^{-1} + \frac{1}{2z}K\left(s\sigma_+ - \frac{1}{2}\sigma_3 \right)K^{-1} \sigma_- \right]\right)_+,
\end{align}
where the notation $(\ldots)_+$ denotes a projection onto the positive power parts of the power series. The above relation gives for the first few $\ell_j$'s,
\begin{align}
&\ell_0(s) = \frac{s}{2},\\
&\ell_1(s) = -4 s r'\left(s^2\right),\\
&\ell_2(s) = \frac{8}{s} \left(r\left(s^2\right){}^2 r'\left(s^2\right)+2
s^2 r'\left(s^2\right){}^2-2 r\left(s^2\right) \left(r'\left(s^2\right)+s^2 r''\left(s^2\right)\right)-2 r_2'\left(s^2\right)\right).
\end{align}
Note that these expressions are independent of $k$.
\end{remark}

\begin{remark}
There has appeared a distinct definition of a Painlev\'{e} III hierarchy in the literature \cite{Sakka}. The model problem for $\Psi$ may be connected with this alternative hierarchy if we choose a different time to act as the independent variable in the ODE. To see this, we generalise the $\Psi$ model problem by altering the behaviour at the origin to
\begin{equation}
\Psi(z)=\Psi_0(s)(I+\bigO(z))e^{-\sum^k_{j=1}\left(-\frac{s_j}{z}\right)^j\sigma_3}z^{\frac{\alpha}{2} \sigma_3}H_j,
\end{equation}
and then define
\beq
\widehat{\Psi}(z,s) := \Psi_0^{-1}(s)\Psi(z^{-1},s).
\eeq
Using the asymptotic behaviour of $\widehat{\Psi}$ in the definition of the Lax pair
\beq
\widehat{B} := \partial_{s_1} \widehat{\Psi} \widehat{\Psi}^{-1} \qquad \mathrm{and} \qquad \widehat{A} := \partial_{z} \widehat{\Psi} \widehat{\Psi}^{-1},
\eeq
we find
\beq
\widehat{B} = \begin{pmatrix}z&\widehat{v}(s_1)\\\widehat{u}(s_1)&z\end{pmatrix},
\eeq
and $\widehat{A} = \sum^{k-1}_{j=-2} A_j(s_1)$, where $A_j$ are $s_1$ dependent matrices and $\widehat{u}$ and $\widehat{v}$ are some undetermined functions. The form of these Lax matrices precisely matches those proposed in \cite{Sakka} as forming a Lax pair for a Painlev\'e III hierarchy. However, in the random matrix model under consideration, it is more natural to work with the variable $s$ instead of $s_1$, and this leads us to a different Painlev\'e hierarchy. 
\end{remark}

\subsubsection*{Double scaling limit for the partition functions in the perturbed LUE and GUE}

Define the pLUE partition function as
\beq\label{partition function pLUE}
Z_{n,k,\alpha}^{pLUE}(t)=\frac{1}{n!}\int_{[0,+\infty)^n}\Delta(x)^2 \prod_{j=1}^n x_j^\alpha e^{-n\left(x_j+\left(\frac{t}{x_j}\right)^k\right)}dx_j,\qquad \alpha>-1,
\eeq
and the pGUE partition function as
\beq\label{partition function pGUE}
Z_{n,k,\alpha}^{pGUE}(t)=\frac{1}{n!}\int_{\mathbb R^n}\Delta(x)^2 \prod_{j=1}^n |x_j|^{2\alpha} e^{-\frac{n}{2}\left(x_j^2+\left(\frac{t}{x^2}\right)^k\right)}dx_j,\qquad \alpha>-1/2.\eeq
Whereas for the pLUE and pGUE partition functions with $k=1$, large $n$ asymptotics are described in terms of a special solution to the PIII equation \cite{BMM, XDZ2}, when perturbing a unitary invariant ensemble with a pole of order $k>1$, the special solutions to the higher order analogues of PIII introduced before will appear. 

\begin{theorem}\label{theorem: partition}
As $n\to\infty$ and $t\to 0$ in such a way that $ s_{n,t}:=2^{-\frac{1}{k}} n^\frac{2k+1}{k} t\to s$, we have the asymptotics
\begin{align}
&\label{as partition function pLUE}Z_{n,k,\alpha}^{pLUE}(t)=Z_{n,k,\alpha}^{pLUE}(0)\exp\left(\frac{1}{2} \int^{s_{n,t}}_0(r(0)-r(\xi))\frac{d\xi}{\xi}\right)\left(1+\bigO(n^{-1}) \right) ,\\
&Z_{n,k,\alpha}^{pGUE}(t)=Z_{n,k,\alpha}^{pGUE}(0)\exp\left(\frac{1}{2} \int^{s_{n,t}}_0(-\alpha^2-r_{\alpha-\frac{1}{2}}(\xi)-r_{\alpha+\frac{1}{2}}(\xi))\frac{d\xi}{\xi}\right)\nonumber\\&\hspace{9cm}\times \ \left(1+\bigO(n^{-1}) \right),\label{as partition function pGUE}
\end{align}
where $r_\alpha(s):=r(s;\alpha)$ is related to the Painlev\'{e} transcendent $y_\alpha(s)$ by 
\[y_\alpha(s) = -2 \frac{d}{ds}\left(r_\alpha(s^2)\right),\qquad r_\alpha(0) = \frac{1}{8}(1-4\alpha^2).\] In other words, we have
\begin{equation}\label{integral r}
r_\alpha(s)=\frac{1}{8}(1-4\alpha^2)-\frac{1}{2}\int_0^{\sqrt{s}}y_\alpha(\eta)d\eta
.
\end{equation}
\end{theorem}

\subsubsection*{Double scaling limit of the correlation kernel}

In what follows, we consider potentials $V$ which are real analytic on $[0,+\infty)$ and such that
\beq
\lim_{x\to +\infty}\frac{V(x)}{\log(x^2+1)}=+\infty.
\eeq

The 
correlation kernel for the eigenvalues in the model (\ref{pUE}) is expressed using orthogonal polynomials with respect to the weight
\beq
w(x) =x^\alpha e^{-n V_k(x)} =x^\alpha\exp\left[-n \left(V(x) + \left(\frac{t}{x}\right)^k\right)\right]
\eeq
on $[0,+\infty)$.
Let
 $p_j$, $j=0,1,\ldots$ be the family of monic polynomials of degree $j$ characterised by the relations
\beq\label{ortho p}
\int^\infty_0 p_j(x) p_m(x) w(x) dx = h_j \delta_{j m}.
\eeq

The correlation kernel of the determinantal point process (\ref{jpdf}) can be written as
\beq
\label{Keqn}
K_n(x,y) = h_{n-1}^{-1}\frac{\sqrt{w(x)w(y)}}{x-y} \left(p_n(x)p_{n-1}(y) - p_n(y)p_{n-1}(x) \right).
\eeq
Important quantities that can be computed directly from $K_n(x,y)$ are the $k$-point correlation functions. The one-point function, or limiting mean eigenvalue density, is given by
\beq
\rho(x) := \lim_{n \rightarrow \infty} \frac{1}{n}K_n(x,x),
\eeq
and describes the macroscopic behaviour of the eigenvalues in the large $n$ limit.
For fixed $t>0$, the measure $d\nu(x)=\rho(x)dx$ is characterised as the equilibrium measure which minimizes
\beq
I_k(\nu) = \iint \log \frac{1}{|x-y|} d\nu(x)d\nu(y) + \int V_k(y) d\nu(y),
\eeq
among all Borel probability measures $\nu$ on $[0,+\infty)$. If we let $n\to\infty$ and at the same time $t\to 0$, the limiting mean eigenvalue density is not affected by the singular perturbation and is equal to the equilibrium measure $d\mu = \psi(x) dx$ obtained from minimising
\beq
I(\mu) = \iint \log \frac{1}{|x-y|} d\mu(x)d\mu(y) + \int V(y) d\mu(y),
\eeq
in which the singular part of $V_k$ has been dropped. It is this measure and density we will use in the remainder of the paper and refer to as the equilibrium measure. 
The measure $\mu$ is characterised by the Euler-Lagrange equations \cite{SaffTotik}, stating that there exists $\ell \in \mathbb{R}$ such that
\begin{align}
&\label{var eq}2\int \log|x-y|\psi(y)dy - V(x) = \ell, \quad x\in\mathrm{supp}\mu, \\
&\label{var ineq}2\int \log|x-y|\psi(y)dy - V(x) \leq \ell, \quad x\in [0,+\infty).
\end{align}


We will require that $V$ is such that the equilibrium measure $\mu$ is supported on a single interval of the form $[0,b]$, 
which implies \cite{DKM}
that its density $\psi$ can be written as
\beq\label{psi in terms of h}
\psi(x) = \frac{1}{2\pi}h(x)\sqrt{\frac{b-x}{x}},\qquad x\in [0,b], 
\eeq where
$h$ is a real analytic function, non-negative on $[0,b]$. We will also require that $V$ is regular, in the sense that
\begin{itemize}
\item[(i)] the density $\psi$ is positive in the interior of its support $(0,b)$, i.e.\ $h(x)>0$ for $0<x<b$,
\item[(ii)] $\psi$ vanishes like a square root at the endpoint $b>0$ and like an inverse square root at the endpoint $0$, i.e.\ $h(b), h(0)\neq 0$,
\item[(iii)] the variational inequality (\ref{var ineq}) is strict on $(b,+\infty)$.
\end{itemize}

Define functions $\psi_j(x;s)$ in terms of the solution to the model RH problem for $\Psi$ as follows.
For $z \in \Omega_1$, let
\beq
\label{psidef}
\begin{pmatrix}\psi_1(z,s)\\ \psi_2 (z,s)\end{pmatrix} := \Psi(z,s) \begin{pmatrix}1\\ 0\end{pmatrix},
\eeq
and for $z \in  \Omega_2 \cup \Omega_3$, we define $\psi_1$ and $\psi_2$ in such a way that they are analytic in $\mathbb C\setminus \Sigma_1$. This means in particular that 
\beq
\label{psidef2}
\begin{pmatrix}\psi_1(x,s)\\ \psi_2 (x,s)\end{pmatrix} = \Psi_+(x,s) \begin{pmatrix}1\\ -e^{-\pi i \alpha}\end{pmatrix}, 
\eeq
for $x<0$.

Then, 
$\psi_1,\psi_2$ satisfy the system of equations
\beq
\partial_z\begin{pmatrix}\psi_1\\ \psi_2 \end{pmatrix} = \frac{1}{s}\begin{pmatrix} \widehat{a} - \frac{1}{2} \widehat{b}s^{-\frac{1}{2}}U(s) & -is^{-\frac{1}{2}} \widehat{b}\\
i s^\frac{1}{2}(\widehat{c}+\widehat{a}s^{-\frac{1}{2}}U(s)- \frac{1}{4} \widehat{b}s^{-1} U(s)^2) & -\widehat{a}+ \frac{1}{2} \widehat{b}s^{-\frac{1}{2}}U(s) \end{pmatrix} \begin{pmatrix}\psi_1\\ \psi_2 \end{pmatrix},
\eeq
and they are characterised as the unique solution with the asymptotic behaviour
\begin{equation}\label{psi as infty}
\begin{pmatrix}\psi_1(z,s)\\ \psi_2(z,s) \end{pmatrix}=\left(I+\bigO(z^{-1})\right)z^{-\frac{1}{4}\sigma_3}Ne^{z^{1/2}\sigma_3},
\end{equation}
as $z \rightarrow \infty$ in $\Omega_1$ (and, more generally, for $|\arg z|<\pi -\epsilon$ for any $\epsilon>0$), and
\begin{equation}\label{psi as 0}
\begin{pmatrix}\psi_1(z,s)\\ \psi_2(z,s) \end{pmatrix}=\bigO(1)e^{-\left(-\frac{s}{z}\right)^k\sigma_3}z^{\frac{\alpha}{2} \sigma_3},
\end{equation}
as $z \rightarrow 0$ for $z \in \Omega_1$.
Here, $U(s) = \int^{s^\frac{1}{2}}_0 y(s') ds'$, $\widehat{a} = a(s^{-1}z,s^\frac{1}{2})$, $\widehat{b} = b(s^{-1}z,s^\frac{1}{2})$, $\widehat{c} = c(s^{-1}z,s^\frac{1}{2})$, with
\begin{align}
&a(z,s) = -\frac{1}{2} \partial_s b(z,s),\\
&c(z,s) = (z-u) b(z,s) - \frac{1}{2} \partial_s^2 b(z,s),
\end{align}
and
\beq
b(z,s) = \frac{4}{(4z)^{k+1}}\sum^k_{n=0} \ell_{k-n}(s) (4z)^n.
\eeq

\begin{theorem}\label{Kthm}In the double scaling limit where $n\to\infty$ and simultaneously $t\to 0$ in such a way that $ 2^{-\frac{1}{k}} c_1 n^\frac{2k+1}{k} t\to s>0$, with $c_1 = b h(0)^2$, we have
\beq
\lim_{n\to\infty}\frac{1}{c_1 n^2}K_n\left(\frac{-u}{c_1 n^2},\frac{-v}{c_1 n^2};t\right) = \mathbb K^{PIII}(u,v;s),
\eeq
for $u,v<0$,
where 
\begin{equation}
\mathbb K^{PIII}(u,v;s)=e^{\pi i\alpha}\frac{\psi_1(u;s)\psi_2(v;s) - \psi_1(v;s)\psi_2(u;s)}{2 \pi i (u-v)}.
\end{equation}
The limit is uniform for $s$ in compact subsets of $(0,+\infty)$.
\end{theorem}
\begin{remark}
For $k=1$, our limiting kernel $\mathbb K^{PIII}(-u,-v;s)$ is equal to the limiting kernel obtained in \cite{XDZ}.
\end{remark}
\begin{theorem}
\label{limitingkernelthm}
The limiting kernel $\mathbb K^{PIII}(u,v;s)$ is positive for $u,v<0$ and $s>0$, and it has the following limits,
\begin{align}
&\lim_{s\to 0}\mathbb K^{PIII}(u,v;s)=\mathbb{J}_\alpha(-u,-v),\\
&\lim_{s\to +\infty} \frac{s^{\frac{2\eta}{3}}}{c_2} K^{PIII}\left(s^{\eta}\left(z_0 + s^{-\frac{\eta}{3}} \frac{u}{c_2}\right),s^{\eta}\left(z_0 + s^{-\frac{\eta}{3}} \frac{v}{c_2}\right);s\right)= \mathbb{A}(u,v),
\end{align}
where $\eta = \frac{2k}{2k+1}$, and  $\mathbb{J}_\alpha(u,v)$ and $\mathbb{A}(u,v)$ are the Bessel and Airy kernels defined as
\begin{align}
&\mathbb{J}_\alpha(u,v)=\frac{J_\alpha(\sqrt{u})\sqrt{v}J_\alpha'(\sqrt{v})-J_\alpha(\sqrt{v})\sqrt{u}J_\alpha'(\sqrt{u})}{2(u-v)},\\
&\mathbb{A}(u,v)=\frac{\Ai(u)\Ai'(v)-\Ai(v)\Ai'(u)}{u-v}.
\end{align}
The constant $c_2$ has the explicit form
\beq
c_2 = \left(\frac{3}{2}\right)^\frac{2}{3}(-z_0)^{-1-\frac{2k}{3}} \sum_{j=0}^{k-1}\frac{(2j+1)!!}{2^j j!}.
\eeq
\end{theorem}

\subsubsection*{Outline}
In Section \ref{section: 2}, we will study in more detail the model RH problem for $\Psi$. Using Lax pair type arguments, we will establish the relation between the RH problem and the PIII hierarchy. In addition, we will prove the solvability of the RH problem for $s>0$, which implies that the PIII solution $y(s)$ has no singularities for $s>0$.
In Section \ref{section: 3}, we will study the large $n$ asymptotics for the orthogonal polynomials associated to our random matrix model by means of the Deift/Zhou steepest descent method. A crucial feature in this analysis will be the construction of a local parametrix near the origin in terms of the model RH problem for $\Psi$.
In Section \ref{section: 4}, we will obtain large $s$ asymptotics for $\Psi$ and for the PIII solution $y$ using an asymptotic analysis of the model RH problem. In Section \ref{section: 5}, small $s$ asymptotics for $\Psi$ and $y$ are obtained.
The results from Section \ref{section: 2} together with the asymptotics from Section \ref{section: 4} and Section \ref{section: 5} will lead to a proof of Theorem \ref{P3theorem}.
In Section \ref{section: 6}, we use the asymptotics for the orthogonal polynomials obtained in Section \ref{section: 3} to prove Theorem \ref{Kthm}.
In Section \ref{section: 7}, we use a differential identity which enables us to derive asymptotics for the pLUE and pGUE partition functions from the asymptotics for the orthogonal polynomials, and to prove Theorem \ref{theorem: partition}.

\section{Properties of the model RH problem}\label{section: 2}
\subsection{Connection to the Painlev\'{e} III hierarchy: proof of Theorem \ref{P3theorem}}

We begin by introducing a new model RH problem for a function $\Phi$, obtained from the model problem for $\Psi$ defined in the introduction.
Let
\beq\label{def Phi}
\Phi(z,s) := \begin{pmatrix}1&0\\s^{-1}r(s^2)&1\end{pmatrix} s^{\frac{\sigma_3}{2}} e^{\frac{1}{4}i \pi \sigma_3} \Psi(s^2 z, s^2),
\eeq
with $r$ defined by (\ref{Psic}).
Using the RH conditions for $\Psi$, we find that $\Phi$ satisfies the RH problem given below.
\subsubsection*{RH problem for $\Phi$}
\begin{itemize}
\item[(a)] $\Phi:\mathbb C\setminus\Sigma\to\mathbb C^{2\times 2}$ analytic.
\item[(b)] $\Phi$ has the jump relations 
\begin{align}
&\Phi_+(z)=\Phi_-(z)\begin{pmatrix}1&0\\-e^{\pi i\alpha}&1\end{pmatrix}, &z\in \Sigma_1,\\
&\Phi_+(z)=\Phi_-(z)\begin{pmatrix}0&-1\\1&0\end{pmatrix}, &z\in \Sigma_2,\\
&\Phi_+(z)=\Phi_-(z)\begin{pmatrix}1&0\\-e^{-\pi i\alpha}&1\end{pmatrix}, &z\in \Sigma_3.
\end{align}
\item[(c)] As $z\to\infty$, $\Phi$ has the asymptotic behaviour
\begin{multline}\label{Phi c}
\Phi(z)=\begin{pmatrix}1&0\\v(s)&1\end{pmatrix} \left(I+\frac{1}{z}\begin{pmatrix}w(s)&v(s)\\h(s)&-w(s)\end{pmatrix} +\bigO(z^{-2})\right)\\
\times e^{\frac{1}{4}i \pi \sigma_3}z^{-\frac{1}{4}\sigma_3}Ne^{sz^{1/2}\sigma_3},
\end{multline}
where $N=\frac{1}{\sqrt{2}}(I+i\sigma_1)$ as before, and
\begin{align}
&w(s) = s^{-2} q(s^2), \\
&v(s) = s^{-1} r(s^2), \\
&h(s) = s^{-3} p(s^2).
\end{align}
\item[(d)] As $z\to 0$, there exists a matrix $\Phi_0(s)$, independent of $z$, such that $\Phi$ has the asymptotic behaviour
\begin{equation}
\label{Phi0}
\Phi(z)=\Phi_0(s)(I+\bigO(z))e^{\frac{(-1)^{k+1}}{z^{k}}\sigma_3}z^{\frac{\alpha}{2} \sigma_3}H_j,
\end{equation}
for $z$ in sector $\Omega_j$, where $H_1, H_2, H_3$ are given in the model problem for $\Psi$, see (\ref{H1})-(\ref{H3}).
\end{itemize}

\noindent Under the assumption that the model RH problem is solvable for $s>0$ (we will prove this assumption in Section \ref{section: solvability}), we are now in a position to prove Theorem \ref{P3theorem}. First we note that the limit in (\ref{qPhi}) exists by a simple application of \eqref{Psic}, and that the differentiability with respect to the parameter $s$ can be proved using general methods which involve re-writing the RH conditions as a singular integral equation.
We have
\beq
\label{qveqn}
y(s) = -2\frac{d}{ds}(r(s^2)) = -2\frac{d}{ds}(s v(s)).
\eeq
For the relation to the PIII hierarchy we use well-known methods of isomonodromic deformation theory. In particular we define the Lax matrices,
\beq\label{Lax AB}
A := \left(\partial_z \Phi \right)\Phi^{-1} \qquad \mathrm{and} \qquad B := \left(\partial_s \Phi \right)\Phi^{-1},
\eeq
and note that, due to the jump matrices of $\Phi$ being constant, the Lax matrices are meromorphic functions in $z$, and may only have isolated singularities. Using the asymptotic expansion of $\Phi$ near $z = \infty$ and $z=0$ one can easily show that $A$ and $B$ are Laurent polynomials of the form
\begin{align}
\label{ALax}
&A(z;s)=\sum_{j=0}^{k+1}A_j(s)z^{-j},\qquad A_0=\frac{s}{2}\sigma_-,\\
\label{BLax}
&B(z;s)=\begin{pmatrix}0&1\\z-u(s)&0\end{pmatrix},
\end{align}
where $u = 2w(s)-v'(s) - v(s)^2$ and $\sigma_-=\begin{pmatrix}0&0\\1&0\end{pmatrix}$. Computing the $(1,2)$-entry of the $z^{-1}$ term in $B$ using the asymptotic behaviour of $\Phi$ at $z = \infty$ together with the fact that $B$ is regular at zero implies 
\beq
\label{wexp}
v'(s)-v(s)^2 + 2w(s) = 0.
\eeq
We can also give explicit expressions for the elements of $A_j$ in terms of coefficients appearing in the asymptotic expansion of $\Phi$ near $z=\infty$. Although this will prove largely unnecessary, one important result obtained in this way is that
\begin{align}
&(A_1)_{12} =\frac{s}{2}, \\
&(A_2)_{12} = s w(s) -\frac{1}{2} s v(s)^2-\frac{1}{2}v(s) = -\frac{1}{2}\frac{d}{ds}\left(s v(s) \right)=y(s), 
\label{A212}
\end{align}
where in the second line we have used \eqref{wexp}. To complete the proof we now show that $u(s)$ is the function appearing in \eqref{P3def}. To this end, consider the compatibility condition,
\beq
\partial_z \partial_s \Phi = \partial_s \partial_z \Phi,
\eeq
which, upon substituting in the expression (\ref{Lax AB}) for the derivatives in terms of the Lax matrices, becomes,
\beq
\label{ZCeqn}
\partial_z B - \partial_s A + [B,A] = 0.
\eeq
If we parameterise $A$ as
\beq\label{para A}
A(z,s) = \begin{pmatrix}a(z,s)&b(z,s)\\c(z,s)&-a(z,s)\end{pmatrix},
\eeq
then \eqref{ZCeqn} is equivalent to three coupled ODEs, 
\begin{align}
\label{asol}
&a(z,s) = -\frac{1}{2} \partial_s b(z,s),\\
\label{csol}
&c(z,s) = (z-u) b(z,s) - \frac{1}{2} \partial_s^2 b(z,s),\\
\label{aceqn}
&\partial_s c(z,s) = 1 + 2(z-u(s))a(z,s).
\end{align}
The first two equations provide $a$ and $c$ in terms of $b$. Although the final equation can be used to obtain an equation for $b$, this is not the course we will follow here, however we will remark upon this possibility later. 

So far we have only used the asymptotic behaviour of $\Phi$ near zero to determine the order of the pole occurring at the origin in each Lax matrix. To complete the proof we now consider the behaviour of $\Phi$ near zero in more detail. Working in $\Omega_1$, see Figure \ref{modelcontour}, write \eqref{Phi0} as
\beq
\Phi(z)=\Phi_0(s)F(z)e^{\Delta(z) \sigma_3},
\eeq
where $F(z) = I + \bigO(z)$ as $z \rightarrow 0$, and $\Delta$ can be written in the general form
\begin{align}
&\label{Deldef}
\Delta(z) := -\Delta_0 \log z + \sum^k_{n=1}\frac{\Delta_n}{n z^n},\\
&\label{Delta2}\Delta_0=-\frac{\alpha}{2}, \ \Delta_k=(-1)^{k+1}k,\ \Delta_1=\ldots=\Delta_{k-1}=0.
\end{align}
Note that $\det A$ can be written as
\beq
\det A = \det\left(\partial_z \Phi \Phi^{-1}  \right) = -\Delta'(z)^2 \det \left(1 + \frac{1}{\Delta'(z)}\sigma_3 F'(z) F(z)^{-1} \right),
\eeq
and that as $z \rightarrow 0$,
\begin{align}
&\Delta'(z) = \bigO(z^{-k-1}),\\
&F'(z)F(z)^{-1} = \bigO(1).
\end{align}
Hence we obtain the estimate
\beq
\det A = -\Delta'(z)^2 + \bigO(z^{-k-1}),
\eeq
as $z\rightarrow 0$. Using \eqref{asol} and \eqref{csol} and rearranging gives
\beq
\label{beqn}
\frac{3}{4}(b')^2 + (z-u(s)) b^2 - \frac{1}{4}(b^2)'' - \left(\frac{d}{dz}\Delta(z)\right)^2 = \bigO(z^{-k-1}),\qquad z\to 0,
\eeq
where primes denote derivatives with respect to $s$. The important observation is that the left hand side
behaves as $\bigO(z^{-2k-2})$ as $z \rightarrow 0$ which after comparing with the right hand side yields $k+1$ equations.
Explicitly, in \eqref{beqn}, let
\beq
\label{btol}
b(z,s) = \frac{4}{(4z)^{k+1}}\sum^k_{j=0} \ell_{k-j}(s) (4z)^j,
\eeq
with $l_0 = s/2$. This gives
\begin{align}
\frac{4}{(4 z)^{2k+2}} \sum^k_{j,m =0} \bigg[& \ell_{k-j}\ell_{k-m} (4z)^{j+m+1} - \\
&\left[(\ell_{k-j} \ell_{k-m})'' - 3\ell_{k-j}' \ell_{k-m}' +4u \ell_{k-j} \ell_{k-m} \right] (4z)^{j+m}\bigg] \nn\\
& = \left(\frac{d}{dz}\Delta(z)\right)^2 + \bigO(z^{-k-1}),\qquad z\to 0. \nn
\end{align}
After some rearrangement and dropping terms of order $\bigO(z^{-k-1})$ the above equation can be written as
\begin{align}
&\frac{4}{(4 z)^{2k+2}} \sum^k_{p=0} \sum^p_{q=0} \left(\ell_{k-p+q+1}\ell_{k-q} - (\ell_{k-p+q} \ell_{k-q})'' + 3\ell_{k-p+q}' \ell_{k-q}' -4u \ell_{k-p+q} \ell_{k-q}
\right)(4z)^p \nn \\ 
&=\left(\frac{d}{dz}\Delta(z)\right)^2 + \bigO(z^{-k-1}),\qquad z\to 0.
\end{align}
Finally, substituting in \eqref{Deldef} and comparing coefficients yields \eqref{P3def}, where $\tau_p$ is a constant taking the form
\beq
\tau_p = 4^{2k-p+1} \sum^p_{q=0} \Delta_{k-p+q}\Delta_{k-q}.
\eeq
In our case, by \eqref{Delta2}, we have \eqref{tau}.
This shows that $u(s)$ appearing in \eqref{BLax} together with the quantities $l_j$ satisfy the system of equations \eqref{P3def}. The relation between $l_1(s)$ and $v(s)$ can be obtained from \eqref{qveqn} and \eqref{A212} together with the fact that $(A_2)_{12} = b_1 = \frac{1}{4}l_1$.

We now prove the statements appearing in Remark \ref{lnremark}. We parameterise the behaviour of $\Phi$ as $z \to \infty$ by
\begin{equation}
\Phi(z)=\begin{pmatrix}1&0\\v(s)&1\end{pmatrix} \left(I+\begin{pmatrix}m_1+m_4&m_2\\m_3&m_1-m_4\end{pmatrix}\right)e^{\frac{1}{4}i \pi \sigma_3}z^{-\frac{1}{4}\sigma_3}Ne^{sz^{1/2}\sigma_3},
\end{equation}
where $m_i$ are formal power series in $z^{-1}$ and we note that this can be written as
\beq
\label{PhiK}
\Phi(z)=K(z,s) e^{\frac{1}{4}i \pi \sigma_3}z^{-\frac{1}{4}\sigma_3}Ne^{sz^{1/2}\sigma_3}.
\eeq
Computing the Lax matrix $B$ using the above expression yields the relation
\beq
z [\sigma_-, K]= \partial_s K - 2v'(s) \sigma_- K - [\sigma_+,K],
\eeq
which upon substituting in $v(s) = r(s^2)/s$ gives \eqref{k3rule} and \eqref{k4rule}. Furthermore the fact that $\det K = 1$ yields the recursion relation \eqref{kdef}. Finally, computing the Lax matrix $A$ using \eqref{PhiK} gives
\beq
A = (\partial_z K) K^{-1} +\frac{s}{2} K \sigma_- K^{-1} + \frac{1}{2z}K\left(s\sigma_+ - \frac{1}{2}\sigma_3 \right)K^{-1},
\eeq
from which the $l_j$ may be extracted by projecting out the $1,2$ element of $A$.
\begin{remark}
Substituting \eqref{asol} and \eqref{csol} into \eqref{aceqn} yields
\beq
z \partial_s b(z,s) = \frac{1}{4} \left(\partial_s^3 b(z,s) + 4 u(s) \partial_s b(z,s) + 2 u'(s) b(z,s) \right) + \frac{1}{2}.
\eeq
We may compute $b(z,s)$ by substituting \eqref{btol} into the above equation,
which yields
\beq
\label{lrec}
\ell'_{j+1} = \ell_j''' + 4u \ell_j' + 2u'\ell_j, 
\eeq
for $j\geq 0$ subject to the conditions $\ell_0 = s/2$ and $\ell_{k+1} = 0$. By iterating the above equation one obtains each $\ell_j$ as differential-integro polynomials in $u$. The final condition $\ell_{k+1} = 0$ yields an ODE for $u$ of order $3k+1$, which is an equivalent way to characterise the $k$-th PIII equation. The reason we prefer \eqref{P3def} is that it yields an ODE for $\ell_1$ of order $2k$ which coincides with the third Painlev\'{e} equation for $k=1$.

Finally, it is interesting to note that \eqref{lrec} is precisely the recursion equation for the Lenard differential polynomials appearing in the definition of the Painlev\'{e} I and II hierarchies. The difference between the Painlev\'{e} I, II, and III hierarchies lies in the initial condition; in the PI case it is $\ell_0 = -4u$, for PII $\ell_0=1/2$, and for PIII $\ell_0 = s/2$. This difference in initial condition has a large effect --- in the case of PI and PII the left-hand-side of \eqref{lrec} is always a total derivative and therefore the recursion relation can be integrated at each order; this is not the case for the PIII hierarchy.
\end{remark}

Recall the definition of $\psi_1,\psi_2$ in (\ref{psidef}).

\begin{corollary}
  The $\psi$-functions associated to $\Psi$, defined in
  \eqref{psidef}, satisfy the ODE system \beq \label{diff psi}
\partial_z\begin{pmatrix}\psi_1\\ \psi_2 \end{pmatrix} = \frac{1}{s}\begin{pmatrix} \widehat{a} - \frac{1}{2} \widehat{b}s^{-\frac{1}{2}}U(s) & -is^{-\frac{1}{2}} \widehat{b}\\
  i s^\frac{1}{2}(\widehat{c}+\widehat{a}s^{-\frac{1}{2}}U(s)- \frac{1}{4}
  \widehat{b}s^{-1} U(s)^2) & -\widehat{a}+ \frac{1}{2}
  \widehat{b}s^{-\frac{1}{2}}U(s) \end{pmatrix} \begin{pmatrix}\psi_1\\
  \psi_2 \end{pmatrix}, \eeq where $U(s) = \int^{s^\frac{1}{2}}_0
y(s') ds'$, $\widehat{a} = a(s^{-1}z,s^\frac{1}{2})$, $\widehat{b} =
b(s^{-1}z,s^\frac{1}{2})$ and $\widehat{c} = c(s^{-1}z,s^\frac{1}{2})$,
and are characterised as the unique pair of solutions to \eqref{diff
  psi} with the asymptotic behaviour \eqref{psi as infty}-\eqref{psi
  as 0}.
\end{corollary}
\begin{proof}
  Substituting \eqref{def Phi} in the first equation of \eqref{Lax
    AB}, and using \eqref{para A}, we obtain \eqref{diff psi}. The
  asymptotic conditions are a consequence of conditions (c)-(d) of the
  RH problem for $\Psi$, and the fact that these determine the
  solution uniquely follows from the general theory for linear systems
  of ODEs.
\end{proof}
\subsection{Proof of existence of $\Psi$}\label{section: solvability}

In order to prove the solvability of the RH problem for $\Psi$, we consider its homogeneous version. This consists of the RH conditions (a), (b), and (d), and has condition (c) replaced by
\begin{equation}\label{Psi0c}
  \Psi(z)=\bigO(z^{-1}) \times z^{-\frac{1}{4}\sigma_3}Ne^{z^{1/2}\sigma_3}, \qquad \mbox{ as $z\to\infty$.}
\end{equation}

\begin{lemma} {\bf (vanishing lemma)}
Let $\Psi_0$ satisfy the RH conditions (a), (b), and (d) in the RH problem for $\Psi$, and the homogeneous asymptotic condition \eqref{Psi0c}. Then, $\Psi_0(z)\equiv 0$.
\end{lemma}
\begin{proof}
This was proved in \cite{XDZ} for $k=1$. The proof from \cite{XDZ} is easily extended to the case $k\in\mathbb N$: equations (2.20) and (2.24) in \cite{XDZ} have to be modified in the natural way, but with 
the function $H(z)$ as defined in equation (2.26) in \cite{XDZ}, it is still true that $H(z)=\bigO(1)$ as $z\to 0$. The remaining part of the proof relies on this fact, but not on a more precise description of the behaviour near $0$, and the value of $k$ is unimportant here.
\end{proof}

For a general class of RH problems, the vanishing lemma (i.e. the fact
that the homogeneous RH problem has only the trivial solution) implies
the solvability of the non-homogeneous RH problem. This follows from
the description of RH problems in terms of singular integral
equations. The existence of a solution of the RH problem is then
equivalent to the bijectivity of a certain operator. This operator is
a Fredholm operator of index zero and because of this, bijectivity
follows from the fact that the kernel of the operator is
trivial. Triviality of the kernel is equivalent to the vanishing
lemma. This procedure works for a fairly general class of RH problems,
as long as the RH problem is equivalent to a RH problem without
singular points. The RH problem for $\Psi$ has a singular point at
$0$, but can be transformed to one without singular points using the
representation (\ref{Psi0better}).  

Define
\begin{equation}
\widehat\Psi(z)=\begin{cases}
\Psi(z),&\mbox{ for $|z|>1$,}\\
\widehat\Psi_0(z),&\mbox{ for $|z|<1$}.
\end{cases}
\end{equation}
Then, $\widehat\Psi$ has no jumps inside the unit circle, but it has a new jump on the unit circle, which we choose with clockwise orientation, given by
\begin{equation}
\widehat\Psi_-^{-1}(z)\widehat\Psi_+(z)=\Psi_0^{-1}(z)\Psi(z)=\begin{pmatrix}1&f_2(z)\\0&1
\end{pmatrix}e^{\left(-\frac{s}{z}\right)^k}z^{\frac{\alpha}{2}\sigma_3}
H_j,
\end{equation}
on the part of the unit circle in $\Omega_j$. The RH problem for $\widehat\Psi$ has no singular points, and the vanishing lemma procedure thus applies to this RH problem, and to the RH problem for $\Psi$.
It follows that the RH problem for $\Psi$ is solvable. For more details on the vanishing lemma and how it implies solvability of the non-homogeneous RH problem, we refer to \cite{DKMVZ2, FIKN, FokasMuganZhou, FokasZhou, IKO}.

\section{Asymptotic analysis of the RH problem for orthogonal polynomials}\label{section: 3}
\label{DZsteep}
The steepest descent analysis we will perform follows closely the
calculations in \cite{Vanlessen2} (corresponding to the case $k=0$)
and in \cite{XDZ} (corresponding to $k=1$) which we will refer to
extensively. The main difference is in the extra condition imposed at
$z=0$. This will enforce us to construct a local parametrix near $0$
which is not built out of Bessel functions, but in terms of our model
RH problem associated to higher order PIII equations instead.
\subsection{Equilibrium measure preliminaries}
As stated in the introduction, a central role in the steepest descent
analysis is played by the equilibrium measure $d\mu(x) = \psi(x)
dx$. For our purposes, we only need that the measure $\mu$ is
characterised by the conditions (\ref{var eq})-(\ref{var ineq}). Under
our assumptions on $V$, namely the fact that it is such that the equilibrium measure $\mu$ is supported on a
single interval $[0,b]$ and that it is regular, by (\ref{psi in terms of h}), the equilibrium density has the form\beq \psi(x) = \frac{1}{2\pi i}h(x) R_+(x) \chi_{(0,b]}(x),
\eeq where \beq R(z) = \left(\frac{z-b}{z}\right)^{\frac{1}{2}}.  \eeq
The right endpoint of the support $b > 0$ is a constant depending on
the potential $V$, $h(x)$ is a real analytic function in $x$ and
$\chi_{[0,b]}$ is the indicator function for the set $[0,b]$. The
principal branch of the square root is taken in $R$, so that $R$ is
analytic in $\mathbb{C}\setminus [0,b]$ and positive for
$z>b$. Furthermore, we have $h(x) > 0$ for $x\in [0,b]$.

In the steepest-descent analysis, the following functions will also
prove useful. Firstly we define the ``$g$-function'' \beq g(z) := \int
\log(z-x) d\mu(x), \eeq where the principal branch of the logarithm is
taken, meaning $g$ is analytic on $\mathbb{C} \setminus (-\infty,b]$.
Under the one-cut regularity assumption which we imposed on $V$, the
$g$-function has a number of properties we will make use of, in
particular
\begin{align}
&g_+(x) + g_-(x) - V(x)-\ell = 0 \qquad x\in[0,b] \label{g1},\\
&g_+(x) + g_-(x) - V(x)-\ell < 0 \qquad x\in (b,+\infty),\label{g2} \\
&g_+(x) - g_-(x) = 2\pi i \int^{b}_x \psi(x) dx.\label{g3}
\end{align}
Let us also define \beq\label{xi} \xi(z) = -\frac{ 1}{2}
\int^z_{b}h(s) R(s) ds \qquad \text{for $z  \in \mathbb{C} \setminus
(-\infty, b]$,} \eeq where the integration path does not cross the real
axis.  By (\ref{g3}) we have, \beq\label{g xi 1} g_+(x) - g_-(x) =
2\xi_+(x), \eeq and using (\ref{g1}), we obtain the identity
\beq\label{g xi 2} 2\xi(z)=2g(z)-V(z)-\ell.  \eeq

\subsection{The RH problem for orthogonal polynomials}
An effective way to characterise orthogonal polynomials with respect
to a weight $w$ is via the well known RH problem due to
Fokas-Its-Kitaev \cite{FokasItsKitaev}. In our case the weight $w(x)$
takes the form \beq w(x) = x^\alpha\exp\left[-n \left(V(x) +
    \left(\frac{t}{x}\right)^k\right)\right], \eeq and the relevant RH
problem is specified below.

\subsubsection*{RH problem for $Y$}
\begin{itemize}
\item[(a)] $Y: \mathbb{C}\setminus [0,+\infty) \rightarrow
  \mathbb{C}^{2 \times 2} $ is analytic.
\item[(b)] The limits of $Y$ as $z$ approaches $(0,+\infty)$ from above and below exist, are continuous on $(0,+\infty)$ and are denoted by $Y_+$ and $Y_-$ respectively. Furthermore they are related by
\beq
Y_+(x) = Y_-(x) \begin{pmatrix}1&w(x)\\0&1\end{pmatrix},\qquad x\in (0,+\infty).
\eeq
\item[(c)] $Y(z) = (I+\bigO(z^{-1}))z^{n\sigma_3}$ as $z \rightarrow \infty$.
\item[(d)] $Y(z) = \bigO\begin{pmatrix}1&1\\1&1\end{pmatrix}$ as $z \rightarrow 0$.
\end{itemize}
This RH problem has a unique solution,
\beq
\label{Ysol}
Y(z) = \begin{pmatrix}p_n(z)&q_n(z)\\ -\frac{2\pi i}{h_{n-1}} p_{n-1}(z)& -\frac{2\pi i}{h_{n-1}} q_{n-1}(z)\end{pmatrix},
\eeq
where $p_j$ is the degree $j$ monic orthogonal polynomial with respect to the weight $w(x)$, $h_j$ is the norm of $p_j$, as in (\ref{ortho p}), and
\beq
q_j(z) := \frac{1}{2\pi i} \int^\infty_0 \frac{p_j(x) w(x)}{x-z} dx.
\eeq

\subsection{First transformation $Y \mapsto T$}
Define
\beq
T(z) = e^{-\frac{n\ell}{2}\sigma_3} Y(z) e^{-n(g(x) - \frac{\ell}{2})\sigma_3} e^{-\frac{n}{2}\left(\frac{t}{z}\right)^k \sigma_3}.
\eeq
The matrix $T$ satisfies the RH problem below.

\subsubsection*{RH problem for $T$}
\begin{itemize}
\item[(a)] $T: \mathbb{C}\setminus [0,+\infty) \rightarrow \mathbb{C}^{2 \times 2} $ is analytic.
\item[(b)] $T$ has the jump relations
\begin{align}
&T_+(x) = T_-(x)\begin{pmatrix}e^{-2n\xi_+}&x^\alpha\\0&e^{-2n\xi_-}\end{pmatrix} , &x\in (0,b),\\
&T_+(x) = T_-(x)\begin{pmatrix}1&x^\alpha e^{n \xi(x)}\\0&1\end{pmatrix} , &x\in (b,\infty).
\end{align}
\item[(c)] $T(z) = I+\bigO(z^{-1})$ as $z \rightarrow \infty$.
\item[(d)] $T(z) = \bigO(1)  e^{-\frac{n}{2} \left(\frac{t}{z}\right)^k \sigma_3}$ as $z \rightarrow 0$.
\end{itemize}

\subsection{Second transformation $T \mapsto S$}
\begin{figure}[t]
\centering 
\includegraphics[scale=0.7]{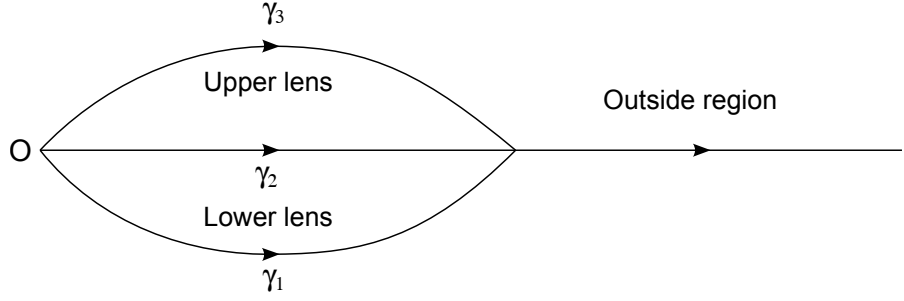}
\caption{The jump contour in the RH problem for $S$.}
\label{lens}
\end{figure}
We now perform the standard technique of opening the lens. This involves introducing a contour $\Sigma_S = (b,\infty) \cup_{i=1}^3 \gamma_i$ as shown in Figure \ref{lens}. As can be seen in the figure the contour divides the complex plane into three disjoint regions. We refer to the finite regions in the upper and lower complex half plane as the upper and lower lens respectively. We then define
\beq
S(z) = \left\{
\begin{array}{lr}
T(z),  & \mbox{for } z \mbox{ outside the lens,}\\
T(z) \begin{pmatrix}1&0\\-z^{-\alpha}e^{-2n\xi(z)}&1\end{pmatrix}, & \mbox{for } z \mbox{ in the upper lens,}\\
T(z) \begin{pmatrix}1&0\\z^{-\alpha}e^{-2n\xi(z)}&1\end{pmatrix}, & \mbox{for } z \mbox{ in the lower lens.}\\
\end{array}
\right.
\eeq
It is then a simple matter to verify that $S$ satisfies the following conditions:
\subsubsection*{RH problem for $S$}
\begin{itemize}
\item[(a)] $S: \mathbb{C}\setminus \Sigma_S \rightarrow \mathbb{C}^{2 \times 2} $ is analytic.
\item[(b)] $S$ has the jump relations
\begin{align}
&S_+(z) = S_-(z)\begin{pmatrix}1&0\\z^{-\alpha} e^{-2n\xi(z)}&1\end{pmatrix} , &z\in \gamma_1 \cup \gamma_3,\\
&S_+(z) = S_-(z)\begin{pmatrix}0&z^{\alpha}\\-z^{-\alpha}&0\end{pmatrix} , &z\in \gamma_2,\\
&S_+(z) = S_-(z)\begin{pmatrix}1&z^\alpha e^{2n \xi(z)}\\0&1\end{pmatrix} , &z\in (b,\infty).
\end{align}
\item[(c)] $S(z) = I+\bigO(z^{-1})$ as $z \rightarrow \infty$.
\item[(d)] The asymptotic behaviour of $S$ near zero is
\beq
\label{Snear0}
S(z) = \bigO(1)  e^{-\frac{n}{2} \left(\frac{t}{z}\right)^k \sigma_3}\times \left\{
\begin{array}{lr}
I,  & \mbox{for } z \mbox{ outside the lens,}\\
\begin{pmatrix}1&0\\-z^{-\alpha}e^{-2n\xi(z)}&1\end{pmatrix}, & \mbox{for } z \mbox{ in the upper lens,}\\
\begin{pmatrix}1&0\\z^{-\alpha}e^{-2n\xi(z)}&1\end{pmatrix}, & \mbox{for } z \mbox{ in the lower lens.}\\
\end{array}
\right.
\eeq
\end{itemize}

\begin{remark}
\label{globalP}
It is important to note that on the contour $\gamma_1\cup \gamma_3\cup (b,\infty)$ away from the origin and $b$, the jump matrices tend to the identity matrix exponentially as $n \rightarrow \infty$, $t\to 0$. This suggests we may approximate $S$ by a global parametrix throughout the complex plane minus small neighbourhoods around the origin and $b$ in which it will be necessary to construct local parametrices. 
\end{remark}

\subsection{The global parametrix: $P^{(\infty)}$}
Following on from Remark \ref{globalP}, a first approximation to $S$ can be obtained by ignoring the entries of the jump matrices which are exponentially suppressed as $n \rightarrow \infty$. We define the global parametrix $P^{(\infty)}(z)$ as a solution to the RH problem obtained in this manner.
\subsubsection*{RH problem for $P^{(\infty)}$}
\begin{itemize}
\item[(a)] $P^{(\infty)}: \mathbb{C} \setminus [0,b] \rightarrow \mathbb{C}^{2 \times 2}$ is analytic.
\item[(b)] $P^{(\infty)}$ has the jump relations
\begin{align}
&P^{(\infty)}_+(z) = P^{(\infty)}_-(z)\begin{pmatrix}0&z^{\alpha}\\-z^{-\alpha}&0\end{pmatrix} , &z\in (0,b).
\end{align}
\item[(c)] $P^{(\infty)}(z) = I+\bigO(z^{-1})\qquad$ as $z \rightarrow \infty$.
\end{itemize}
Note that this RH problem does not define a solution uniquely, since nothing is imposed on the behaviour of $P^{(\infty)}$ near $0$ and $b$.
A solution to the above RH problem can be obtained by defining
\beq
P^{(\infty)}(z) = 2^{-\alpha \sigma_3} b^{\frac{\alpha \sigma_3}{2}} N^{-1} \left(\frac{z-b}{z}\right)^{-\frac{\sigma_3}{4}} N D\left(b^{-1} z\right)^{-\sigma_3} b^{-\frac{\alpha \sigma_3}{2}},
\eeq
where $N = 2^{-\frac{1}{2}}(I + i \sigma_1)$, and where $D$ is the Szeg\H{o} function 
\beq
D(z) := \frac{z^{\frac{\alpha}{2}}}{\varphi(2z-1)^{\frac{\alpha}{2}}} \mbox{,} \qquad  \varphi(z) := z + \sqrt{z^2-1}, 
\eeq
in which the branch is chosen such that $\varphi(z) \sim 2z$ as $z \rightarrow \infty$ and with cut on $[-1,1]$. 

\subsection{Local parametrix near $b$}
Given that the only difference between the RH problem for $S$ appearing here and that of \cite{Vanlessen2} is the local condition (d) \eqref{Snear0} near the hard edge $0$, we see that the local parametrix required in a small disk $U^{(b)}$ around the soft edge $b$ will satisfy exactly the same RH problem (up to a rescaling by $b$) appearing in \cite[Section 3.6]{Vanlessen2}. As is shown in \cite[Section 3.6]{Vanlessen2}, and is by now very familiar, the local parametrix near the soft-edge at $b$ can be built using the Airy model problem, and it satisfies the following RH problem.

\subsubsection*{RH problem for $P^{(b)}$}
\begin{itemize}
\item[(a)] $P^{(b)}: U^{(b)}\setminus \Sigma_S \rightarrow \mathbb{C}^{2 \times 2}$ is analytic.
\item[(b)] $P^{(b)}$ has the same jumps as $S$ restricted to the contours $U^{(b)} \cap \Sigma_S$.
\item[(c)] $P^{(b)}(z) = P^{(\infty)}(z)(I+\bigO(n^{-1}))$ as $n \rightarrow \infty$ uniformly for $z \in \partial U^{(b)}$.
\item[(d)] $S(z)P^{(b)}(z)^{-1}$ is analytic at $z=b$.
\end{itemize}

\subsection{Local parametrix near the origin}
We also require a local approximation to $S$ in a neighbourhood of the origin. This will take the form of a local parametrix $P^{(0)}$ valid within a disk $U^{(0)}$ of fixed size, characterised by the following RH conditions.
\subsubsection*{RH problem for $P^{(0)}$}
\begin{itemize}
\item[(a)] $P^{(0)}(z): U^{(0)}\setminus \Sigma_S \rightarrow \mathbb{C}^{2 \times 2}$ is analytic.
\item[(b)] $P^{(0)}$ has the same jumps as $S$ restricted to the contours $U^{(0)} \cap \Sigma_S$.
\item[(c)] $P^{(0)}(z) = P^{(\infty)}(z)(I+\bigO(n^{-1}))$ uniformly for $z \in \partial U^{(0)}$, in the double scaling limit where $n \rightarrow \infty$ and simultaneously $t\to 0$ in such a way that $2^{-1/k}c_1n^{\frac{2k+1}{k}}t\to s\in (0,+\infty)$.
\item[(d)] As $z \rightarrow 0$,
\beq
P^{(0)}(z) = \bigO(1) e^{-\frac{n}{2} \left(\frac{t}{z}\right)^k\sigma_3} \left\{
\begin{array}{lr}
1  & \mbox{for } z \mbox{ outside the lens}\\
\begin{pmatrix}1&0\\-z^{-\alpha}e^{-2n\xi(z)}&1\end{pmatrix} & \mbox{for } z \mbox{ inside the upper lens}\\
\begin{pmatrix}1&0\\z^{-\alpha}e^{-2n\xi(z)}&1\end{pmatrix} & \mbox{for } z \mbox{ inside the lower lens}\\
\end{array}
\right.
\eeq
\end{itemize}
A convenient way to construct $P^{(0)}$ is to first transform the problem to one with constant jump matrices. This can be achieved by the transformation
\beq\label{def P0hat}
\widehat{P}^{(0)}(z) := P^{(0)}(z) (-z)^{\frac{1}{2} \alpha \sigma_3} e^{n \xi(z) \sigma_3}.
\eeq
Applying this transformation, we obtain a RH problem for $\widehat{P}^{(0)}$.
\subsubsection*{RH problem for $\widehat{P}^{(0)}$} \label{Phatprob}
\begin{itemize}
\item[(a)] $\widehat{P}^{(0)}: U^{(0)} \setminus \Sigma_S \rightarrow \mathbb{C}^{2 \times 2}$ is analytic. This follows from the properties of $P^{(0)}$ together with the fact that $\xi_+(x) - \xi_-(x) = 2 \pi i$ for $x \in (-\infty, 0)$.
\item[(b)] $\widehat{P}^{(0)}$ has the jump relations,
\begin{align}
&\widehat{P}_+^{(0)}(z)=\widehat{P}_-^{(0)}(z)\begin{pmatrix}0&1\\-1&0\end{pmatrix}, &z\in \gamma_2,\\
&\widehat{P}_+^{(0)}(z)=\widehat{P}_-^{(0)}(z)\begin{pmatrix}1&0\\e^{\pi i\alpha}&1\end{pmatrix}, &z\in  \gamma_1,\\
&\widehat{P}_+^{(0)}(z)=\widehat{P}_-^{(0)}(z)\begin{pmatrix}1&0\\e^{-\pi i\alpha}&1\end{pmatrix}, &z\in  \gamma_3.
\end{align}
\item[(c)] $\widehat{P}^{(0)}(z) = P^{(\infty)}(z)(I+\bigO(n^{-1}))(-z)^{\frac{1}{2} \alpha \sigma_3} e^{n \xi(z) \sigma_3}$ uniformly for $z \in \partial U^{(0)}$, in the double scaling limit where $n \rightarrow \infty$ and simultaneously $t\to 0$ in such a way that $2^{-1/k}c_1n^{\frac{2k+1}{k}}t\to s\in (0,+\infty)$.
\item[(d)] As $z \rightarrow 0$,
\beq
\widehat{P}^{(0)}(z) = \bigO(1) e^{-\frac{n}{2} \left(\frac{t}{z}\right)^k\sigma_3} (-z)^{\frac{1}{2} \alpha \sigma_3} \left\{
\begin{array}{lr}
I,  & \mbox{for } z \mbox{ outside the lens,}\\
\begin{pmatrix}1&0\\-e^{-i \pi \alpha}&1\end{pmatrix}, & \mbox{for } z \mbox{ inside the upper lens,}\\
\begin{pmatrix}1&0\\e^{i \pi \alpha}&1\end{pmatrix}, & \mbox{for } z \mbox{ inside the lower lens.}\\
\end{array}
\right.
\eeq\end{itemize}

We now make the connection to the model problem for $\Psi(\zeta,s)$ by asserting that the function $\widehat{P}^{(0)}$ can be expressed as
\beq
\label{P0exp}
\widehat{P}^{(0)}(z,t) = E(z) \Psi\left(n^2 f(z), 2^{-\frac{1}{k}} c_1 n^\frac{2k+1}{k} t\right) \sigma_3 ,
\eeq
where $c_1 =bh(0)^2>0$ and $E(z)$ and $f(z)$ are functions which we will now define. We take $f$ to be
\beq
f(z) = \frac{1}{4}\left(\int_0^zh(s)R(s)ds\right)^2.
\eeq
Then $f$ is a conformal map from a neighbourhood of $0$ to a neighbourhood of $0$, with $f'(0)=-c_1$, and we have $e^{n f(z)^{\frac{1}{2}}} = (-1)^n e^{n\xi(z)}$. This follows from the definition (\ref{xi}) of $\xi$ together with the fact that $\int^{b}_0 \psi(x) dx = 1$.
Let \begin{equation}\label{def snt}
s_{n,t}=2^{-\frac{1}{k}} c_1 n^\frac{2k+1}{k} t,
\end{equation}
and assume that $t$ depends on $n$ in such a way that $s_{n,t}$ lies in a compact subset of $(0,+\infty)$ for $n$ sufficiently large.

It can be checked in a straightforward way that \eqref{P0exp} satisfies conditions (a), (b) and (d) of the RH problem for $\widehat P_0$. To satisfy property (c), we take
\beq
E(z) = (-1)^n P^{(\infty)}(z) \sigma_3 (-z)^{\frac{1}{2} \alpha \sigma_3}N^{-1} \left(n^2 f(z)\right)^\frac{\sigma_3}{4}.
\eeq
Using the jump relation for $P^{(\infty)}$ and choosing the branch cuts for $\left(n^2 f(z)\right)^\frac{\sigma_3}{4}$ and $(-z)^{\frac{1}{2} \alpha \sigma_3}$ on the positive half-line, one shows that $E(z)$ is indeed analytic in $U^{(0)}\setminus\{0\}$, with a singularity at $0$ which is removable. See \cite{Vanlessen2} for further details.
For $z\in \partial U^{(0)}$, we have that $n^2f(z)$ is large, and using the asymptotics (\ref{Psic}) of $\Psi$ at infinity, which are uniform for $s$ in compact subsets of $(0,+\infty)$, one easily obtains condition (c) of the RH problem for $P$.

\subsection{Final transformation: small norm RH problem}
Define
\begin{equation}
R(z)=\begin{cases}
S(z)P^{(\infty)}(z)^{-1},&\mbox{ for $z\in \mathbb C \setminus \left(\overline{U^{(0)} \cup U^{(b)} \cup \Sigma_S}\right)$, }\\
S(z)P^{(0)}(z)^{-1},&\mbox{ for $z\in U^{(0)}$,}\\
S(z)P^{(b)}(z)^{-1},&\mbox{ for $z\in U^{(b)}$.}
\end{cases}
\end{equation}
Then $R$ satisfies the following RH problem.
\subsubsection*{RH problem for $R$}
\begin{itemize}
\item[(a)] $R: \mathbb C\setminus \Sigma_R\rightarrow \mathbb{C}^{2 \times 2}$, with
\[\Sigma_R=\partial U^{(0)}\cup \partial U^{(b)}\cup \left(\Sigma_S\setminus(U^{(0)}\cup U^{(b)}\cup[0,b])\right)\] is analytic, with the circles around $0$ and $b$ oriented clockwise. 
\item[(b)] $R$ satisfies the jump relation $R_+(z)=R_-(z)J_R(z)$ for $z\in\Sigma_R$, with
\begin{equation}\label{JR}
J_R(z)=
\begin{cases}
P^{(0)}(z)P^{(\infty)}(z)^{-1},& \mbox{ for $z\in\partial U^{(0)}$},\\
P^{(b)}(z)P^{(\infty)}(z)^{-1},& \mbox{ for $z\in\partial U^{(b)}$},\\
P^{(\infty)}(z)J_S(z)P^{(\infty)}(z)^{-1},& \mbox{ for $z\in\Sigma_R\setminus\left(\partial U^{(0)}\cup \partial U^{(b)}\right)$,}
\end{cases}
\end{equation}
where $J_S$ denotes the jump matrix for $S$.
\item[(c)] As $z\to\infty$, we have
\begin{equation}\label{as R z}
R(z)=I+\frac{R_1}{z}+\bigO(z^{-2}),
\end{equation}
where $R_1$ is independent of $z$ but depends on $n$.\end{itemize}
The above RH conditions are straightforward to verify using the properties of $S$, $P^{(\infty)}$, $P^{(0)}$, and $P^{(b)}$. For condition (a), one needs to verify in particular that $S(z)P^{(0)}(z)^{-1}$ is analytic at $0$. To see this, one needs to make use of the local behaviour of $\Psi$ at the origin, see (\ref{Psi0}) or (\ref{Psi0better}), and of the jump relations for $S$ and $P^{(0)}$.

From the matching conditions for $P^{(0)}$ and $P^{(b)}$ with $P^{(\infty)}$, and the fact that the jump matrix for $S$ is exponentially close to the identity on $\Sigma_R\setminus\left(\partial U^{(0)}\cup \partial U^{(b)}\right)$ as $n\to\infty$, we observe that
\begin{equation}\label{as JR}
J_R(z)=I+\bigO(n^{-1}),\qquad z\in\Sigma_R, n\to\infty.
\end{equation}
This estimate holds in the double scaling limit where $n \rightarrow \infty$ and simultaneously $t\to 0$ in such a way that $2^{-1/k}c_1n^{\frac{2k+1}{k}}t\to s\in (0,+\infty)$.
By the general theory for small-norm RH
problems~\cite[Section 7]{DKMVZ1}, this implies that
\begin{equation}\label{as R}
R(z)=I+\bigO(n^{-1}),\qquad n\to\infty,
\end{equation}
uniformly for $z$ off the jump contour $\Sigma_R$ in the double scaling limit.

\subsection{Computation of the error}
The estimate (\ref{as R}) will be sufficient for the proof of Theorem \ref{Kthm}. For the proof of Theorem \ref{theorem: partition}, we need more precise asymptotics for $R$. To that end, let us be more precise in (\ref{as JR}). In the remaining part of this section, we will specialise to the perturbed Laguerre case where $V(x)=x$. Then, $b=1$ and $h(z)=1$. The estimates below are valid in the double scaling limit where $n \rightarrow \infty$ and simultaneously $t\to 0$ in such a way that $2^{-1/k}c_1n^{\frac{2k+1}{k}}t\to s\in (0,+\infty)$.

 We can write 
\begin{equation}\label{JR as2}
J_R(z)=I+\frac{1}{n}J_1(z)+\bigO(n^{-2}),\qquad n\to\infty,
\end{equation}
where the matrix $J_1(z)$ is non-zero only on $\partial U^{(0)}$ and $\partial U^{(1)}$.
We will express the first correction term $R^{(1)}$ in the large $n$ expansion for $R$, given by
\begin{equation}\label{R large n}
R(z)=I+R^{(1)}(z)n^{-1}+\bigO(n^{-2}),\qquad z\in\mathbb C\setminus\Sigma_R,\ n\to\infty,
\end{equation}
in terms of $J_1$.
 On $\partial U^{(0)}$, we can use (\ref{def P0hat}), (\ref{P0exp}), and the asymptotics (\ref{Psic}) for $\Psi$ to conclude that
\begin{equation}\label{J1 0}
J_1(z)=\frac{r}{2f(z)^{1/2}}P^{(\infty)}(z)
\begin{pmatrix}1&ie^{-\pi i\alpha}z^\alpha\\ie^{\pi i\alpha}z^{-\alpha}&-1\end{pmatrix} P^{(\infty)}(z)^{-1},
\end{equation}
for
$z\in\partial U^{(0)}$. 
On the Airy disk, $J_1$ can also be computed using the explicit construction of the Airy parametrix and the asymptotic behaviour of the Airy function for large arguments.  After a lengthy computation which we do not include here, we obtain \begin{equation}\label{J1 b}
J_1(z)=\frac{1}{24\widetilde f(z)^{3/2}}P^{(\infty)}(z)
\begin{pmatrix}1&6i z^\alpha\\6iz^{-\alpha}&-1\end{pmatrix} P^{(\infty)}(z)^{-1},
\end{equation}
for $z\in\partial U^{(1)}$, where $\widetilde f$ is a conformal map sending $1$ to $0$, with $\widetilde f'(1)>0$.
Substituting the large $n$ expansion (\ref{JR as2}) in the jump relation $R_+=R_-J_R$, and noting that $R(z)\to I$ as $z\to\infty$, we obtain
\subsubsection*{RH problem for $R^{(1)}$}
\begin{itemize}
\item[(a)] $R^{(1)}:\mathbb C\setminus\left(\partial U^{(0)}\cup \partial U^{(1)}\right)\to\mathbb C^{2\times 2}$ is analytic,
\item[(b)] $R_+^{(1)}(z)=R_-^{(1)}(z)+J_1(z)$ for $z\in\partial U^{(0)}\cup \partial U^{(1)}$,
\item[(c)] $R^{(1)}(z)\to 0$ as $z\to\infty$. 
\end{itemize}
This scalar additive RH problem for $R^{(1)}$ can be solved explicitly, and the solution is given by
\begin{equation}
R^{(1)}(z)=\frac{1}{2\pi i}\int_{\partial U^{(0)}\cup \partial U^{(1)}}\frac{J_1(\xi)}{\xi-z}d\xi,
\end{equation}
where the integral is taken entry-wise.
By (\ref{as R z}) and (\ref{R large n}), we obtain
\begin{equation}
R_1=-\frac{1}{2\pi i n}\int_{\partial U^{(0)}\cup \partial U^{(1)}}J_1(\xi)d\xi + \bigO(n^{-2}),\qquad n\to\infty,
\end{equation}
and the term of order $n^{-1}$ in this expression can be easily computed using residue calculus. In the Laguerre case where $V(x)=x$, we have as $n\to\infty$,  
\begin{equation}
\label{LUER}
R_{1}=\frac{1}{n}\begin{pmatrix}
\frac{1}{16}-\frac{\alpha^2}{4}&\frac{i}{48}(12\alpha^2+24\alpha+11)
\\
\frac{i}{48}(12\alpha^2-24\alpha+11)
& -\frac{1}{16}+\frac{\alpha^2}{4}\end{pmatrix}
-\frac{r}{2n}\begin{pmatrix}1&i\\i&-1
\end{pmatrix}+\bigO(n^{-2}).
\end{equation}
The first term comes from the Airy parametrix, the second term is the contribution of the parametrix near $0$, and contains the function $r$ related to the $k$-th higher order PIII equation.

\section{Asymptotic analysis for $\Psi$ as $s\to +\infty$}\label{section: 4}
In this section we compute the asymptotic behaviour for $\Psi(z,s)$ as $s\rightarrow \infty$. As a corollary,
we also obtain the asymptotic behaviour as $s \rightarrow \infty$ of the relevant solution to the $k$-th PIII equation. Our approach in this section is based on a standard steepest-descent analysis.

\subsection{Re-scaling of the model problem}
We begin by performing a transformation $Y \mapsto T$ which rescales the variable $z$ in a convenient way. We define
\begin{equation}
T(z;s)=s^{\frac{k}{4k+2}\sigma_3}
\Psi(s^{\frac{2k}{2k+1}}z;s).
\end{equation}
Then we have the following RH problem for $T$.
\subsubsection*{RH problem for $T$}
\begin{itemize}
\item[(a)] $T:\mathbb C\setminus\Sigma\to\mathbb C^{2\times 2}$ is analytic, with $\Sigma$ the jump contour for $\Psi$. 
\item[(b)] $T$ has the same jump relations as $\Psi$, 
\begin{align}
&T_+(z)=T_-(z)\begin{pmatrix}0&-1\\1&0\end{pmatrix}, &z\in \Sigma_2,\\
&T_+(z)=T_-(z)\begin{pmatrix}1&0\\-e^{\pi i\alpha}&1\end{pmatrix}, &z\in \Sigma_1,\\
&T_+(z)=T_-(z)\begin{pmatrix}1&0\\-e^{-\pi i\alpha}&1\end{pmatrix}, &z\in \Sigma_3.
\end{align}
\item[(c)] As $z\to\infty$, $T$ has the asymptotic behaviour
\begin{equation}\label{T c}
T(z)=\left(I+\frac{T_1(s)}{z} +\bigO(z^{-2})\right)z^{-\frac{1}{4}\sigma_3}Ne^{s^{\frac{k}{2k+1}}z^{1/2}\sigma_3},
\end{equation}
where 
\begin{equation}\label{C1hat}
T_1(s)=s^{-\frac{2k}{2k+1}} s^{\frac{k}{4k+2}\sigma_3} \begin{pmatrix}q(s)&-i r(s)\\i p(s)&-q(s)\end{pmatrix}s^{-\frac{k}{4k+2}\sigma_3}.
\end{equation}
\item[(d)] As $z\to 0$, there exists a matrix $T_0(s)$, independent of $z$, such that $T$ has the asymptotic behaviour
\begin{equation}
\label{T0}
T(z)=T_0(s)(I+\bigO(z))e^{(-1)^{k+1}\frac{s^{\frac{k}{2k+1}}}{z^{k}}\sigma_3}z^{\frac{\alpha}{2} \sigma_3}H_j,
\end{equation}
for $z \in \Omega_j$.
\end{itemize}

\subsection{Contour deformation}
\begin{figure}[t]
\centering 
\includegraphics[scale=0.6]{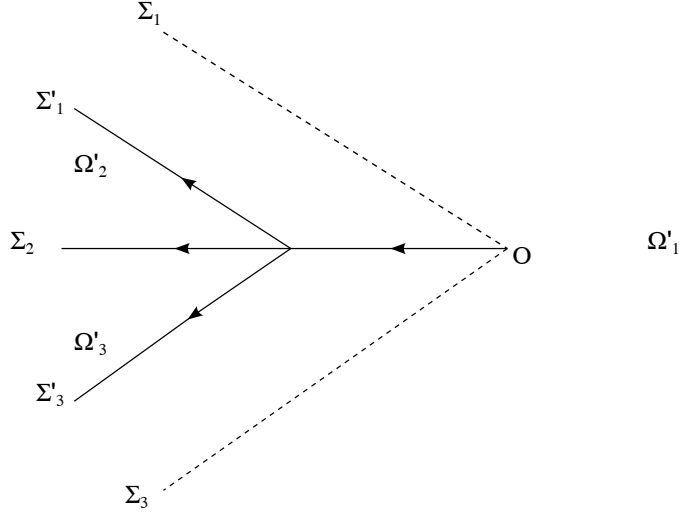}
\caption{The deformed jump contours appearing in the Riemann-Hilbert problem for $S$. The dashed lines denote the location of the original contours $\Sigma_1$ and $\Sigma_3$.}
\label{contoursS}
\end{figure}

We now note that we can deform the contours $\Sigma_1$ and $\Sigma_3$ to any contour in the upper/lower half plane. In particular we deform $\Sigma_1$ and $\Sigma_3$ to the union of $(z_0,0)$ and $\Sigma'_1$ and $\Sigma'_3$ respectively, where $z_0<0$ is a fixed point which we specify later, and where $\Sigma'_1, \Sigma'_3$ are shown in Figure \ref{contoursS}.
Recall that the angle between $\Sigma_1',\Sigma_3'$ (as for $\Sigma_1,\Sigma_3$) and the real line can be chosen freely between $0$ and $\pi$.

 We denote the collection of contours by $\Sigma'_S = \mathbb R^- \cup \Sigma'_1\cup \Sigma'_3$.
We see from the figure that this deformation results in new regions $\Omega'_j$. For $z$ in region $\Omega'_j$ for $j=1,2,3$, we define $S$ to be the analytic continuation of $T$ from region $\Omega_j$ to region $\Omega'_j$.

\subsubsection*{RH problem for $S$}
\begin{itemize}
\item[(a)] $S:\mathbb C\setminus\Sigma'_S\to\mathbb C^{2\times 2}$ is analytic.
\item[(b)] $S$ has the jump relations 
\begin{align}
&S_+(z)=S_-(z)\begin{pmatrix}0&-1\\1&0\end{pmatrix}, &z\in (-\infty,z_0),\\
&S_+(z)=S_-(z)\begin{pmatrix}1&0\\-e^{\pi i\alpha}&1\end{pmatrix}, &z\in \Sigma'_1,\\
&S_+(z)=S_-(z)\begin{pmatrix}1&0\\-e^{-\pi i\alpha}&1\end{pmatrix}, &z\in \Sigma'_3,\\
&S_+(z)=S_-(z)\begin{pmatrix}e^{-\pi i\alpha}&-1\\0&e^{\pi i\alpha}\end{pmatrix}, &z\in (z_0,0).
\end{align}
\item[(c)] As $z\to\infty$, $S$ has the asymptotic behaviour
\begin{equation}\label{S c}
S(z)=\left(I+\frac{T_1(s)}{z} +\bigO(z^{-2})\right)z^{-\frac{1}{4}\sigma_3}Ne^{s^{\frac{k}{2k+1}}z^{1/2}\sigma_3}.
\end{equation}
\item[(d)] As $z\to 0$,
\begin{equation}
\label{S0}
S(z)=T_0(s)(I+\bigO(z))e^{(-1)^{k+1}\frac{s^{\frac{2k}{2k+1}}}{z^{k}}\sigma_3}z^{\frac{\alpha}{2} \sigma_3}.
\end{equation}
\end{itemize}
Note that the only remaining sector around the origin is now $\Omega'_1$, and this is the reason why $H_j$ in RH condition (d) is replaced by $H_1=I$.

\subsection{The $g$-function}
We define the $g$-function as follows,
\begin{equation}
\label{gdef}
g(z)=(z-z_0)^{3/2}p_{k-1}(z)z^{-k},
\end{equation}
where $(z-z_0)^{3/2}$ is defined to be analytic on $\mathbb{C} \setminus (\infty, z_0]$ and positive for $z>z_0$, and where $p_{k-1}$ is a monic polynomial of degree $k-1$ fixed by the condition
\begin{equation}
g(z)=(-1)^{k+1}z^{-k}+\bigO(1),\qquad z\to 0.
\end{equation}
This fixes the polynomial $p_{k-1}$ and the constant $z_0$, we have
\begin{align}
&p_{k-1}(z)=z^{k-1} + \sum_{j=0}^{k-2}\beta_j z^j,\\
&\beta_j=(-1)^{j+k-1}(-z_0)^{-\frac{3}{2}-j}\frac{(2j+1)!!}{2^j j!},\\
&z_0=- \left(\frac{2^{k-1}(k-1)!}{(2k-1)!!}\right)^{-\frac{2}{2k+1}},
\end{align}
where $(2j+1)!!=(2j+1)(2j-1)\ldots 3 . 1$. Note that we have
\begin{equation}\label{ginfty}
g(z)=z^{\frac{1}{2}}+ g_1 z^{-1/2}+ g_2 z^{-3/2}+\bigO(z^{-5/2}),\qquad z\to\infty,
\end{equation}
where
\begin{align}\label{g1b}
g_1 &= \beta_{k-2} -\frac{3}{2} z_0, \\
g_2 &= \frac{3}{8}z_0^2 + \beta_{k-3} -\frac{3}{2} z_0 \beta_{k-2}.
\end{align}
We end this subsection by proving the following proposition which will be needed to show that the jump matrices for $U$ (defined below) tend to $I$ as $s\to +\infty$, for $z$ away from the origin.
\begin{proposition}\label{prop g}
\begin{itemize}
\item[(i)] For $z_0 < z < 0$, we have $\Re(g(z)) < 0$.
\item[(ii)] There exists a $\theta \in (0,\pi)$ such that for $\theta<|\arg(z-z_0)| <\pi$, $\Re(g(z)) > 0$.
\end{itemize}
\end{proposition}
\begin{proof}
We rewrite \eqref{gdef} as
\beq
g(z) = -(-z_0)^{-k}\left(\frac{z}{z_0}\right)^{-k}\left(1-\frac{z}{z_0}\right)^\frac{3}{2}\widetilde{p}_{k-1}\left(\frac{z}{z_0}\right),
\eeq
where
\beq
\widetilde{p}_{k-1}(z) := \sum^{k-1}_{j=0} \frac{(2j+1)!!}{2^j j!} z^j.
\eeq
Now note that for $0 < z/z_0 < 1$, $g(z) < 0$, thereby proving the first part of the proposition. The second part follows from the Cauchy-Riemann conditions if we demonstrate that $\Im g_+'(z) < 0$ for $z < z_0$. To this end we compute
\beq
g'(z) = (-z_0)^{-k-1} \left(\frac{z}{z_0}\right)^{-k-1} \left(1-\frac{z}{z_0}\right)^\frac{1}{2} \widetilde{q}_{k-1}\left(\frac{z}{z_0}\right),
\eeq
where
\beq
\widetilde{q}_{k-1}(z) := -k \left(1 + \sum^{k}_{j=1} \frac{(2j-1)!!}{2^j j!} z^j \right).
\eeq
Due to the negative coefficients in $\widetilde{q}$, we have indeed that $\Im g_+'(z) < 0$ for $z/z_0 > 1$, from which the proposition follows.
\end{proof}

\subsection{Normalisation at $\infty$ and $0$ of the RH problem}
We now use the $g$-function to normalise the behaviour of the RH solution near the origin and at infinity. Define
\begin{equation}
U(z;s)=(I+i s^{\frac{k}{2k+1}}g_1 \sigma_-)S(z;s)e^{-s^{\frac{2k}{2k+1}}g(z)\sigma_3}.
\end{equation}
This transformation leaves the jump contours unchanged and gives us the following RH problem for $U$.
\subsubsection*{RH problem for $U$}
\begin{itemize}
\item[(a)] $U:\mathbb C\setminus\Sigma'_S\to\mathbb C^{2\times 2}$ is analytic.
\item[(b)] $U$ has the jump relations 
\begin{align}
&U_+(z)=U_-(z)\begin{pmatrix}0&-1\\1&0\end{pmatrix}, &z\in (-\infty,z_0),\\
&U_+(z)=U_-(z)\begin{pmatrix}1&0\\-e^{\pi i\alpha}e^{-2s^{-\frac{k}{2k+1}}g(z)}&1\end{pmatrix}, &z\in \Sigma'_1,\\
&U_+(z)=U_-(z)\begin{pmatrix}1&0\\-e^{-\pi i\alpha}e^{-2s^{-\frac{k}{2k+1}}g(z)}&1\end{pmatrix}, &z\in \Sigma'_3,\\
&U_+(z)=U_-(z)\begin{pmatrix}e^{-\pi i\alpha}&-e^{2s^{-\frac{k}{2k+1}}g(z)}\\0&e^{\pi i\alpha}\end{pmatrix}, &z\in (z_0,0).
\end{align}
\item[(c)] As $z\to\infty$, $U$ has the asymptotic behaviour 
\begin{equation}\label{U c}
U(z)=\left(I+\frac{U_1(s)}{z}+\bigO(z^{-2})\right)z^{-\frac{1}{4}\sigma_3}N,
\end{equation}
with
\begin{align}
U_1(s) = &(I+i s^{\frac{k}{2k+1}}g_1 \sigma_-) \times \\
&\bigg(T_1(s)(I-i s^{\frac{k}{2k+1}}g_1 \sigma_-) + \nn \\
&\frac{1}{2}s^\frac{2k}{2k+1}g_1^2 I + i s^\frac{k}{2k+1}g_1 \sigma_+ -i(s^\frac{k}{2k+1}g_2+\frac{1}{6} s^\frac{3k}{2k+1} g_1^3)\sigma_- \bigg). \nn
\end{align}
\item[(d)] As $z\to 0$,
\begin{equation}
\label{U0}
U(z)=U_0(s)(I+\bigO(z))z^{\frac{\alpha}{2} \sigma_3}.
\end{equation}
\end{itemize}
Note that we have
\begin{equation}\label{C112tilde}
(U_1)_{12}(s) = i s^\frac{k}{2k+1}g_1(s) + (T_1)_{12}(s),
\end{equation}
and hence, by (\ref{C1hat}),
\begin{equation}\label{v in C1tilde}
 r(s)= i s^\frac{k}{2k+1} (U_1)_{12}(s) + s^\frac{2k}{2k+1}g_1(s).
\end{equation}

By Proposition \ref{prop g}, if we choose $\Sigma_1', \Sigma_2'$ close enough to $(-\infty, z_0)$, we have that the jump matrices for $U$  converge to the identity as $s\to +\infty$ on the part of the jump contour away from the real line. The off-diagonal entry of the jump matrix on $(z_0,0)$ also tends to $0$. The convergence is not uniform near $z_0$, and therefore we need to construct a local Airy parametrix near $z_0$ to model the jumps near $z_0$, and a global parametrix to model the jumps on the real line.

\subsection{Global parametrix}

Ignoring exponentially small jumps and a small neighbourhood of $z_0$, we have the following RH problem which we will solve explicitly.

\subsubsection*{RH problem for $P^{(\infty)}$}
\begin{itemize}
\item[(a)] $P^{(\infty)}:\mathbb C\setminus\mathbb R^-\to\mathbb C^{2\times 2}$ is analytic.
\item[(b)] $P^{(\infty)}$ has the jump relations 
\begin{align}
&P^{(\infty)}_+(z)=P^{(\infty)}_-(z)\begin{pmatrix}0&-1\\1&0\end{pmatrix}, &z\in (-\infty,z_0),\\
&P^{(\infty)}_+(z)=P^{(\infty)}_-(z)e^{-\pi i\alpha\sigma_3}, &z\in (z_0,0).
\end{align}
\item[(c)] As $z\to\infty$, $P^{(\infty)}$ has the asymptotic behaviour
\begin{equation}\label{Pinf c}
P^{(\infty)}(z)=\left(I+P_1^{(\infty)}z^{-1}+\bigO(z^{-2})\right)z^{-\frac{1}{4}\sigma_3}N.
\end{equation}
\item[(d)] As $z\to 0$, $P^{(\infty)}$ has the asymptotic behaviour
\begin{equation}\label{P0 c}
P^{(\infty)}(z)=\bigO(1)z^{\frac{\alpha}{2} \sigma_3}.
\end{equation}
\end{itemize}
The solution to this RH problem can be constructed from the one in \cite{XDZ} and is given by
\beq\label{Pinf}
P^{(\infty)}(z)=\begin{pmatrix}1&0\\i\alpha (-z_0)^\frac{1}{2}&1 \end{pmatrix}(z-z_0)^{-\frac{\sigma_3}{4}}N\left(\frac{\sqrt{z-z_0}+(-z_0)^\frac{1}{2}}{\sqrt{z-z_0}-(-z_0)^\frac{1}{2}} \right)^{-\frac{\alpha}{2}\sigma_3}.
\eeq

\subsection{Local parametrix near $z_0$}
Using the model RH problem for the Airy function, a solution to the following RH problem can be constructed in a fixed size disk $D$ around $z_0$.
\subsubsection*{RH problem for $P^{(z_0)}$}
\begin{itemize}
\item[(a)]$P^{(z_0)}$ is analytic in $D\setminus\Sigma'_S$.
\item[(b)] For $z\in D\cap \Sigma'_S$, $P^{(z_0)}$ satisfies exactly the same jump relations as $U$.
\item[(c)] For $z\in \partial D$, $P^{(z_0)}$ matches with the global parametrix,
\begin{equation}
P^{(z_0)}(z)=P^{(\infty)}(z)(I+\bigO(s^{-\frac{k}{2k+1}})),\qquad s\to +\infty.
\end{equation}
\item[(d)] $U(z)P^{(z_0)}(z)^{-1}$ is analytic at $z_0$.
\end{itemize}
The construction of $P^{(z_0)}$ proceeds in the standard way. We first introduce the matrix-valued function
\begin{equation}\label{definition: Psi}
    \Upsilon(z) := M_A \times
    \begin{cases}
        \begin{pmatrix}
            \Ai( z) & \Ai(\omega^2 z) \\
            \Ai'( z) & \omega^2\Ai'(\omega^2 z)
        \end{pmatrix}e^{-\frac{\pi i}{6}\sigma_3}, & \mbox{for $ z\in\Omega_{A,1}$,} \\[3ex]
        \begin{pmatrix}
            \Ai( z) & \Ai(\omega^2 z) \\
            \Ai'( z) & \omega^2\Ai'(\omega^2 z)
        \end{pmatrix}e^{-\frac{\pi i}{6}\sigma_3}
        \begin{pmatrix}
            1 & 0 \\
            -1 & 1
        \end{pmatrix}, & \mbox{for $ z\in\Omega_{A,2}$,} \\[3ex]
        \begin{pmatrix}
            \Ai( z) & -\omega^2\Ai(\omega z) \\
            \Ai'( z) & -\Ai'(\omega z)
        \end{pmatrix}e^{-\frac{\pi i}{6}\sigma_3}
        \begin{pmatrix}
            1 & 0 \\
            1 & 1
        \end{pmatrix}, & \mbox{for $ z\in\Omega_{A,3}$,} \\[3ex]
        \begin{pmatrix}
            \Ai( z) & -\omega^2\Ai(\omega z) \\
            \Ai'( z) & -\Ai'(\omega z)
        \end{pmatrix}e^{-\frac{\pi i}{6}\sigma_3}, & \mbox{for $ z\in\Omega_{A,4}$,}
    \end{cases}
\end{equation}
with $\omega=e^{\frac{2\pi i}{3}}$, $\Ai$ the Airy function,
\beq
M_A := \sqrt{2\pi} e^{\frac{\pi i}{6}} \begin{pmatrix}1 &0\\0& -i \end{pmatrix},
\eeq
and the regions $\Omega_{A,i}$ defined as in Figure \ref{airycontours}. The function $\Upsilon$ solves the following RH problem.

\begin{figure}[t]
\centering 
\includegraphics[scale=0.6]{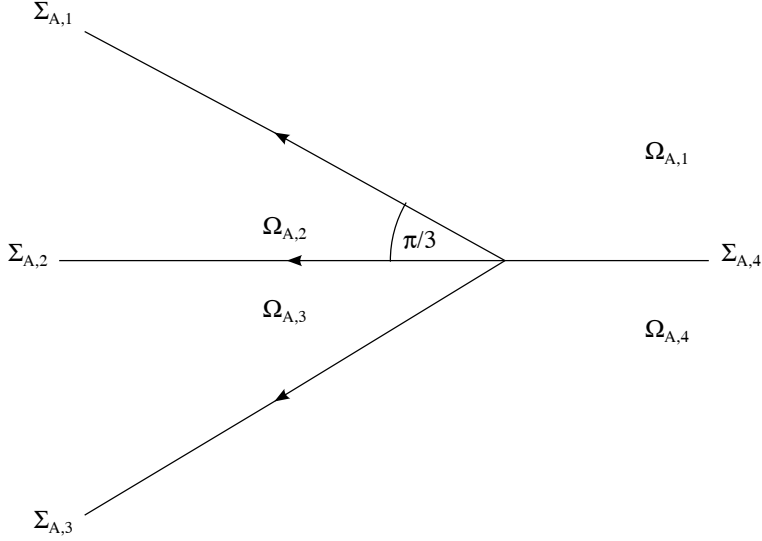}
\caption{The jump contour $\Sigma_A$ for the Airy model problem.}
\label{airycontours}
\end{figure}

\subsubsection*{Airy model RH problem}
\begin{itemize}
    \item[(a)] $\Upsilon:\mathbb{C}\setminus \Sigma_A \to\mathbb{C}^{2\times
    2}$ is analytic where $\Sigma_A = \cup^4_{i=1} \Sigma_{A,i}$.
\item[(b)] $\Upsilon$ has the jump relations 
\begin{align}
&\Upsilon_+(z)=\Upsilon_-(z)\begin{pmatrix}1&0\\-1&1\end{pmatrix}, &z\in \Sigma_{A,1},\\
&\Upsilon_+(z)=\Upsilon_-(z)\begin{pmatrix}0&-1\\1&0\end{pmatrix}, &z\in \Sigma_{A,2},\\
&\Upsilon_+(z)=\Upsilon_-(z)\begin{pmatrix}1&0\\-1&1\end{pmatrix}, &z\in \Sigma_{A,3},\\
&\Upsilon_+(z)=\Upsilon_-(z)\begin{pmatrix}1&-1\\0&1\end{pmatrix}, &z\in \Sigma_{A,4}.
\end{align}

    \item[(c)] As $z\rightarrow \infty$, we have
    \beq\label{asymptotics: Upsilon}
        \Upsilon( z) =  z^{-\frac{\sigma_3}{4}}N\left(I+\bigO(z^{-\frac{3}{2}} \right)e^{-\frac{2}{3} z^{3/2}\sigma_3}.
    \eeq
\end{itemize}

\noindent Using the Airy model problem we can then write $P^{(z_0)}$ as
\beq
\label{Pz0}
P^{(z_0)}(z) = E_1(z) \Upsilon(s^\frac{2k}{6k+3} f_1(z)) e^{-s^\frac{k}{2k+1} g(z) \sigma_3} \theta(z),
\eeq
where $\theta(z) = e^{\pm\frac{1}{2}\pi i \alpha \sigma_3}$ for $\pm \Im z >0$, and $f_1$ is given by
\beq
f_1(z) := \frac{3}{2} (-z_0)^{-\frac{2k}{3}}\left(1-\frac{z}{z_0} \right) \left[\left(\frac{z}{z_0}\right)^{-k}\widetilde{p}_{k-1}\left(\frac{z}{z_0}\right) \right]^\frac{2}{3}.
\eeq

Note that the specific form of $f_1$ has been chosen to satisfy $-\frac{2}{3} (s^\frac{2k}{6k+3} f)^\frac{3}{2} = s^\frac{k}{2k+1} g$. Furthermore, using properties of $\widetilde{p}_{k-1}$ it is straightforward to show that $f_1$ is conformal in $D$ when $D$ is sufficiently small, with $f_1'(z_0)>0$. Finally, we let
\beq
E_1(z) = P^{(\infty)}(z) \theta(z)^{-1} N^{-1} (s^\frac{2k}{6k+3} f_1(z))^\frac{\sigma_3}{4}.
\eeq
Again it is straightforward to demonstrate that $E_1$ is analytic in $D$ thus showing that \eqref{Pz0} possesses the correct jumps. The specific form of $E(z)$ is chosen as to satisfy the boundary condition (c) in the RH problem for $P^{(z_0)}$.

\subsection{Small norm RH problem}
Define
\begin{equation}
R(z)=\begin{cases}
U(z)P^{(\infty)}(z)^{-1},&\mbox{ for $z\in \mathbb C\setminus\overline{D}$,}\\
U(z)P^{(z_0)}(z)^{-1},&\mbox{ for $z\in D$.}
\end{cases}
\end{equation}
Then $R$ is the solution to a RH problem normalised at infinity and with jump matrices uniformly close to the identity as $s\to +\infty$. One shows that
\begin{equation}
R(z,s)=I+\bigO(s^{-\frac{k}{2k+1}}),\qquad \frac{d}{ds}R(z,s)=\bigO(s^{-1-\frac{k}{2k+1}}),
\end{equation}
uniformly for $z$ off the jump contour for $R$, as $s\to +\infty$.
Defining $R_1$ in terms of the large $z$ expansion of $R$,
\begin{equation}\label{def R1}
R(z)=I+\frac{R_1}{z}+\bigO(z^{-2}),
\end{equation}
we have in particular that $R_1(s)=\bigO(s^{-\frac{k}{2k+1}})$.

\subsection{Asymptotics for $r(s)$ and $y(s)$}
For $z$ large, we have
\begin{equation}
U(z)=R(z)P^{(\infty)}(z).
\end{equation}
Comparing the large $z$ expansions (\ref{U c}), (\ref{Pinf c}), and (\ref{def R1}) for $U$, $P^{(\infty)}$, and $R$, we obtain
\begin{equation}
U_1(s)=P_1^\infty(s)+R_1(s)=P_1^\infty(s)+\bigO(s^{-\frac{k}{2k+1}}),\qquad s\to\infty.
\end{equation}
It follows from (\ref{v in C1tilde}), (\ref{g1b}), and the large $z$ expansion of $P^{(\infty)}$ given by (\ref{Pinf}) that
\begin{equation}
r(s)=s^\frac{2k}{2k+1} (\beta_{k-2} -\frac{3}{2} z_0) - (-z_0)^\frac{1}{2} s^\frac{k}{2k+1} \alpha + \bigO(1),
\end{equation}
as $s\to\infty$. By (\ref{qveqn}), we obtain (\ref{as y infty}). 

\section{Asymptotic analysis for $\Psi$ as $s\to 0$}\label{section: 5}
In this section, we compute the asymptotic behaviour of $\Psi(z,s)$ as $s\rightarrow 0$. As a corollary,
we obtain the asymptotic behaviour as $s \rightarrow 0$ of the associated solution to the $k$-th Painlev\'e III equation.

\subsection{The Bessel model RH problem}
If we set $s=0$ in the RH problem for $\Psi$, we obtain a RH problem which is equivalent to the Bessel model RH problem used in \cite{Vanlessen2, XDZ}. This suggests that the solution to the Bessel model problem will be a good approximation to $\Psi$. For $z$ close to the origin, this is not true anymore because of the singular behaviour of $\Psi$ near $0$ for $s\neq 0$, which is not present for $s=0$. This is the reason why we will need to construct a local parametrix near the origin.
\subsubsection*{Bessel RH for $\Upsilon$}
\begin{itemize}
\item[(a)] $\Upsilon:\mathbb C\setminus\Sigma\to\mathbb C^{2\times 2}$ analytic, with 
\[\Sigma=\Sigma_1\cup\Sigma_2\cup\Sigma_3\] as in Figure \ref{modelcontour}.
\item[(b)] $\Upsilon$ has the jump relations 
\begin{align}
&\Upsilon_+(z)=\Upsilon_-(z)\begin{pmatrix}0&-1\\1&0\end{pmatrix}, &z\in \Sigma_2,\\
&\Upsilon_+(z)=\Upsilon_-(z)\begin{pmatrix}1&0\\-e^{\pi i\alpha}&1\end{pmatrix}, &z\in \Sigma_1,\\
&\Upsilon_+(z)=\Upsilon_-(z)\begin{pmatrix}1&0\\-e^{-\pi i\alpha}&1\end{pmatrix}, &z\in \Sigma_3.
\end{align}
\item[(c)] As $z\to\infty$, $\Upsilon$ has the asymptotic behaviour
\begin{equation}\label{Upsilonc}
\Upsilon(z)=\left(I+\bigO(z^{-1})\right)z^{-\frac{1}{4}\sigma_3}Ne^{z^{1/2}\sigma_3},
\end{equation}
where $N=\frac{1}{\sqrt{2}}(I+i\sigma_1)$.
\item[(d)] Conditions as $z \rightarrow 0$;
\begin{itemize}
\item[(i)] If $\alpha < 0$  then $\Upsilon(z) = \bigO \begin{pmatrix} |z|^\frac{\alpha}{2} & |z|^\frac{\alpha}{2}\\ |z|^\frac{\alpha}{2}&|z|^\frac{\alpha}{2} \end{pmatrix} $.
\item[(ii)] If $\alpha  = 0$ then $\Upsilon(z) = \bigO \begin{pmatrix} \log|z| &\log|z|\\ \log|z|& \log|z| \end{pmatrix} $.
\item[(iii)] If $\alpha > 0$ then 
\beq
\Upsilon(z) = \begin{cases}\bigO  \begin{pmatrix} |z|^\frac{\alpha}{2} & |z|^{-\frac{\alpha}{2}}\\ |z|^\frac{\alpha}{2}&|z|^{-\frac{\alpha}{2}} \end{pmatrix},&\mbox{ for $z\in\Omega_1$,} \\
\bigO\begin{pmatrix} |z|^{-\frac{\alpha}{2}} & |z|^{-\frac{\alpha}{2}}\\ |z|^{-\frac{\alpha}{2}}&|z|^{-\frac{\alpha}{2}} \end{pmatrix} ,&\mbox{ for $z\in\Omega_2\cup\Omega_3$.}
\end{cases}
\eeq
\end{itemize}
\end{itemize}
The solution to the above RH problem can be constructed using Bessel functions \cite{Vanlessen2,XDZ}, we have
\beq
\Upsilon(z) = \left(I + \frac{i}{8}(4\alpha^2 + 3) \sigma_-\right) \pi^\frac{\sigma_3}{2} \begin{pmatrix}I_\alpha(z^\frac{1}{2})& \frac{i}{\pi} K_\alpha(z^\frac{1}{2})\\ \pi i z^\frac{1}{2}I_\alpha'(z^\frac{1}{2})& -z^\frac{1}{2} K_\alpha'(z^\frac{1}{2}) \end{pmatrix} H_j.\label{Upsilon Bessel}
\eeq
when $z \in \Omega_j$. In the construction of the local parametrix for $\Psi$ near the origin, we will also require more detailed knowledge of the behaviour of $\Upsilon$ at the origin. It can be verified that for $z \in \Omega_j$ and $\alpha \notin \mathbb{Z}$, we have
\beq\label{repr Upsilon origin1}
\Upsilon(z) = \Upsilon_0(z) z^\frac{\alpha \sigma_3}{2} \begin{pmatrix} 1 & \frac{1}{2i\sin \pi\alpha} \\ 0 & 1\end{pmatrix} H_j,
\eeq
whereas for $\alpha \in \mathbb{Z}$, we have
\beq\label{repr Upsilon origin2}
\Upsilon(z) = \Upsilon_0(z) z^\frac{\alpha \sigma_3}{2} \begin{pmatrix} 1 & \frac{i}{\pi}(-1)^{\alpha+1} \log \frac{\sqrt{z}}{2}\\ 0 & 1\end{pmatrix} H_j.
\eeq
In both cases $\Upsilon_0(z)$ is an entire function.

\subsection{The local parametrix near the origin}
Around the origin, $\Upsilon$ ceases to be a good approximation to $\Psi$ due to the singular behaviour of $\Psi$. We therefore construct a local parametrix $F$ around the origin in a disk of the form $D=\{z\in\mathbb C:|z|<\epsilon\}$ with $\epsilon>0$ fixed but sufficiently small. The construction done here deviates from the construction done in \cite{XDZ}. We prefer this modified construction because it is more natural given the behaviour of $\Psi$ near $0$ given in (\ref{Psi0better}).
\subsubsection*{RH problem for $F$}
\begin{itemize}
\item[(a)] $F:D\setminus\Sigma\to\mathbb C^{2\times 2}$ is analytic.
\item[(b)] $F$ has the same jumps as $\Psi$ on $\Sigma$.
\item[(c)] On the circle $|z| = \epsilon$ we have the following matching conditions as $s \rightarrow 0$,
\beq F(z,s) =
\begin{cases} \Upsilon(z)(I + \bigO(s^k))&\mbox{ for $\alpha \notin \mathbb{Z}$}, \\
\Upsilon(z)(I + \bigO(s^k)),& \mbox{ for $\alpha \in \mathbb{Z}$}.
\end{cases}
\eeq
\item[(d)] The function $\Psi F^{-1}$ is bounded near $z=0$.
\end{itemize}
The solution to this problem takes the following form,
\beq\label{def F}
F(z,s) = \Upsilon_0(z)
\begin{pmatrix}1&h(z)\\0&1\end{pmatrix}\begin{pmatrix}1 & f_2(z,s)\\ 0 & 1 \end{pmatrix} z^\frac{\alpha}{2} e^{-(-\frac{s}{z})^k\sigma_3} H_j,
\eeq
where $f_2(z,s)$ is given in (\ref{def f origin}), $\Upsilon_0$ in (\ref{repr Upsilon origin1})-(\ref{repr Upsilon origin2}), $H_j$ in (\ref{H1})-(\ref{H3}), and 
\begin{align}\label{def h}
&h(z)=-f_2(z,0)+\frac{1}{2i\sin \pi\alpha}z^\alpha,& \mbox{ if $\alpha\notin\mathbb Z$,}\\
&h(z)=-f_2(z,0)+\frac{i}{\pi}\log\frac{\sqrt{z}}{2},&\mbox{ if $\alpha\in\mathbb Z$.}
\end{align}
It is easily verified that $h$ is an entire function.

Conditions (a) and (b) of the RH problem are direct, since $F$ is of the form
\[\mbox{Analytic function } \times \begin{pmatrix}1 & f_2(z,s)\\ 0 & 1 \end{pmatrix} z^\frac{\alpha}{2} e^{-(-\frac{s}{z})^k} H_j,\]
just like $\Psi$, see (\ref{Psi0better}), and therefore $F$ has precisely the same jump conditions as $\Psi$.
To verify condition (c), note that for $|z|=\epsilon$, $\alpha\notin\mathbb Z$, by (\ref{repr Upsilon origin1}),
\begin{align}
&F(z)\Upsilon(z)^{-1}\nonumber\\
&=\Upsilon_0(z)\begin{pmatrix}1&h(z)+f_2(z,s)\\0&1\end{pmatrix}e^{-(-\frac{s}{z})^k\sigma_3}
\begin{pmatrix} 1 & -\frac{z^\alpha}{2i\sin \pi\alpha} \\ 0 & 1\end{pmatrix} \Upsilon_0(z)^{-1}\nonumber\\
&=I+\bigO(s^k),\qquad\mbox{ as $s\to 0$.}
\end{align}
Condition (d) follows from (\ref{Psi0better}) and (\ref{def F}).
In the case $\alpha\in\mathbb Z$, the RH conditions are checked in a similar way.

\subsection{Small norm RH problem}
We can now construct a small norm RH problem. Define
\begin{equation}
R(z,s)=\begin{cases}
\Psi(z,s)\Upsilon^{-1}(z),&\mbox{ for $|z| > \epsilon$,}\\
\Psi(z,s)F(z,s)^{-1},&\mbox{ for $|z| < \epsilon$.}
\end{cases}
\end{equation}
Then $R$ is a small-norm RH problem and we have
\begin{equation}\label{error R}
R(z,s)=I + \bigO(s^k),\qquad s\to 0,
\end{equation}
and
\begin{equation}\label{error R der}
\frac{d}{ds}R(z,s)=\bigO(s^{k-1}),\qquad s\to 0,
\end{equation}
uniformly for $z$ off the jump contour for $R$.

\subsection{Asymptotics for $r(s)$ as $s\to 0$}
To obtain the initial condition for $r$, we compute $(\Psi e^{-z^\frac{1}{2} \sigma_3} N^{-1} z^\frac{\sigma_3}{4})_{12}$ as $z \rightarrow \infty$ in two different ways; firstly using \eqref{Psic} and secondly using the small $s$ asymptotics obtained above. This gives
\beq
r(s) = \frac{1}{8}\left(1-4\alpha^2 \right)+\bigO(s^k),\qquad r'(s)=\bigO(s^{k-1}), \qquad s\to 0, \quad s>0,
\eeq
valid for any $\alpha > -1$. By (\ref{qveqn}), (\ref{error R}), and (\ref{error R der}), we obtain (\ref{as y 0}).

\begin{remark}\label{remark uniform0}
Recall the asymptotic condition (\ref{Psic}) for $\Psi(z,s)$ as $z\to\infty$. We used before that this condition holds uniformly for $s$ in compact subsets of $(0,+\infty)$. As a consequence of the asymptotic analysis as $s\to 0$ done in this section, it follows that this expansion is uniform even for $s$ small. This implies that the matching condition between the global parametrix and the local parametrix near $0$ in the asymptotic analysis for the orthogonal polynomials holds also for $s_{n,t}$ small, and that the error term $\bigO(n^{-2})$ for $R$ in (\ref{R large n}) is uniform also as $n\to\infty$ and $t\to 0$ in such a way that $2^{-1/k}c_1n^{\frac{2k+1}{k}}t\to 0$. We will use this fact later on. 
\end{remark}

\section{Asymptotics for the kernel near the origin}\label{section: 6}
The correlation kernel \eqref{Keqn} can be expressed in terms of the matrix $Y$ as
\beq
\label{KYeqn}
K_n(x,y;t) = \frac{1}{2\pi i} \frac{\sqrt{w(x)w(y)}}{x-y} \begin{pmatrix}0&1\end{pmatrix} Y_+(y,t)^{-1}Y_+(x,t) \begin{pmatrix}1\\0\end{pmatrix}.
\eeq
To prove Theorem \ref{Kthm}, we use the large $n$ asymptotics obtained from the steepest descent analysis in Section
\ref{DZsteep}. By inverting the sequence of transformations $Y \mapsto T \mapsto S \mapsto R$, we may express $Y(z)$ for $z$ in the upper lens region as
\beq
Y(z) = e^{\frac{1}{2}nl\sigma_3}  R(z)E(z) \Psi\left(n^2 f(z), 2^{-\frac{1}{k}} c_1 n^\frac{2k+1}{k} t\right) e^{\frac{i \pi (\alpha-1)}{2}\sigma_3} \begin{pmatrix}1&0\\1&1\end{pmatrix} w(z)^{-\frac{1}{2}\sigma_3}.
\eeq
Inserting this into \eqref{KYeqn}, we obtain
\begin{multline}
K_n(x,y;t) = \frac{1}{2\pi i(x-y)}(-1,-1)e^{\frac{-i \pi \alpha}{2}\sigma_3}\Psi_+\left(n^2 f(y),s_{n,t}\right)^{-1} L^{-1}(y) \\
\times \ 
L(x)\Psi_+\left(n^2 f(x),s_{n,t}\right) e^{\frac{i \pi \alpha }{2}\sigma_3}(1,-1)^T, \label{Kexp}
\end{multline}
where $L(z) := R(z)E(z)$ and $s_{n,t}$ given by (\ref{def snt}). The boundary value $\Psi_+$ is taken according to the oriented contour $\Sigma$ in Figure \ref{modelcontour}, which means that we take the limit from region $\Omega_3$. We now zoom in on the origin by defining $x = -c_1^{-1} n^{-2}u$ and $y = -c_1^{-1} n^{-2}v$, with (as before) $c_1=-f'(0)$, and consider the  asymptotics of $K_n(x,y;t)$ in the double scaling limit where $n\to\infty$, $t\to 0$ in such a way that $s_{n,t}$ remains bounded. First we note that 
\beq
\label{Lasymp}
L^{-1}(y)L(x) = I + (u-v)\bigO(n^{-2}),\qquad n\to\infty.
\eeq
Additionally we have
\beq
\label{fasymp}
n^2f(x) = u\left(1+ \bigO(n^{-2})\right),\qquad n\to\infty,
\eeq
and similarly for $n^2f(y)$.
If $s_{n,t}\to s$ as $n\to\infty, t\to 0$, with $s\in (0,+\infty)$, we 
have \beq\label{Psi as}\Psi_+\left(n^2 f(y),s_{n,t}\right)\to \Psi_+(v,s),\qquad \Psi_+\left(n^2 f(x),s_{n,t}\right)\to \Psi_+(u,s).\eeq
Substituting \eqref{Lasymp}, \eqref{fasymp}, and \eqref{Psi as} into \eqref{Kexp}, and moreover expressing $\Psi_+$ in terms of the functions $\psi_1$ and $\psi_2$ by (\ref{psidef2}), we obtain
\beq
\lim_{n\to\infty}\frac{1}{c_1 n^2}K_n\left(\frac{u}{c_1 n^2},\frac{v}{c_1 n^2}; \frac{s}{c_2 n^{2+1/k}}\right) = e^{\pi i\alpha} \frac{\psi_1(u;s)\psi_2(v;s) - \psi_1(v;s)\psi_2(u;s)}{2 \pi i (u-v)},
\eeq
after a straightforward calculation. This proves Theorem \ref{Kthm}.

Finally, using the asymptotics of $\Psi(z,s)$ as $s\to \infty$ and $s \to 0$, we can prove Theorem \ref{limitingkernelthm}. In the case $s \to 0$, we have 
\beq
\Psi_+(z,s) = (I + \bigO(s^k)) \Upsilon_+(z)  \begin{pmatrix}1\\-e^{-\pi i \alpha}\end{pmatrix},
\eeq
uniformly for $z$ bounded away from zero. Now we can substitute the above expression and the expression (\ref{Upsilon Bessel}) for $\Upsilon$ in terms of Bessel functions into \eqref{psidef2}. Using the resulting expression for $\psi_i$ together with the definition of $\mathbb{K}^{PIII}$, we obtain, in a similar way as in \cite[Section 5.2]{XDZ}, the first part of the theorem.

 The limit of the kernel as $s\to \infty$ is more complicated; one must invert the sequence of transformations performed in Section \ref{section: 4} which mapped the RH problem for $\Psi$ to the Airy model RH problem. The result of this is that, for $z$ close to $z_0$, we have the $s \to \infty$ asymptotics
\begin{align}
\Psi(z,s) =& s^{-\frac{1}{4}\eta\sigma_3}\left(I - i s^\frac{\eta}{2}g_1 \sigma_-\right)R(s^{-\eta}z,s)E_1(s^{-\eta}z,s)\times \nn \\
&\Upsilon(s^\frac{\eta}{3}f_1(s^{-\eta}z))\theta(z)
\begin{cases}I &\mbox{ for $z \in  \Omega_1\cup\Omega'_2\cup\Omega'_3$,}\\
\begin{pmatrix}1&0\\e^{-i\pi\alpha}&1\end{pmatrix} &\mbox{ for $z \in  \Omega'_1\cap\Omega_3$,}\\
\begin{pmatrix}1&0\\-e^{i\pi\alpha}&1\end{pmatrix} &\mbox{ for $z \in  \Omega'_1\cap\Omega_2$.}
\end{cases}
\end{align}
where $\eta = 2k/(2k+1)$. Using the above expression in \eqref{psidef2} we can compute asymptotics for $\psi_i$ as $s \to \infty$ in the regions $\Omega'_3$ and $\Omega'_1\cap\Omega_3$, and in particular on the negative real line. Substituting these asymptotics into the definition of $\mathbb{K}^{PIII}$, we obtain the second part of the theorem.
\section{Asymptotics for the partition function}\label{section: 7}

In this section, we will prove Theorem \ref{theorem: partition} by deriving large $n$ asymptotics for the partition functions in the perturbed LUE and perturbed GUE. 

\subsection{The partition function in the pLUE}
The main ingredient required to obtain asymptotics for $Z_{n,k}^{pLUE}$ is a differential identity with respect to the perturbation parameter $t$. More general identities of the form below were obtained in \cite{BEH06,BMM}, but for completeness we give a proof of the differential identity relevant to us.

We note that the function $z \Tr \left(Y^{-1}(z)Y'(z)\right)$ is not analytic at infinity: indeed it has a discontinuity on the real line. However, since the discontinuity becomes exponentially small as $z\to\pm\infty$,  $z \Tr \left(Y^{-1}(z)Y'(z)\right)$ does admit a full asymptotic expansion in negative powers of $z$ as $z\to\infty$.
We have
\begin{equation}
z \Tr \left(Y^{-1}(z)Y'(z)\right)=\frac{c_1}{z}+\bigO(z^{-2}),
\end{equation}
and denote $\Res_{z=\infty}\left(z \Tr \left(Y^{-1}(z)Y'(z)\right)\right):=-c_1$, although strictly speaking this is not the residue at infinity, because infinity is not an isolated singularity.
The integral \[\oint z \Tr \left(Y^{-1}(z)Y'(z)  \sigma_3 \right)dz\]
over a large counter-clockwise oriented circle $\{z:|z|=M\}$ is not exactly equal to
$2\pi i c_1$, but it tends to $2\pi ic_1$ as the radius $M$ tends to infinity. This observation will be important in the proof of the following lemma.

\begin{lemma}[cf.\ \cite{BEH06,BMM}]
\label{lem:dif_id}
We have
  \begin{equation}
    \label{eq:dif_id}
    \frac{d}{dt} \log Z_{n,k,\alpha}^{pLUE}(t) = \frac{n^2 + \alpha n}{t} + \frac{n}{2t} \Res_{z=\infty}\left(z \Tr \left(Y^{-1}(z)Y'(z)  \sigma_3 \right)\right).
\end{equation}
\end{lemma}
\begin{proof}
Recall \eqref{partition function}, which, after the substitution $x_i = t \xi_i$, becomes \begin{equation}Z_{n,k,\alpha}^{pLUE} = t^{n^2 + \alpha n} \widetilde{Z}_{n,k,\alpha}^{pLUE},\label{part sub}\end{equation} where
\beq
\widetilde{Z}_{n,k,\alpha}^{pLUE} := \frac{1}{n!}\int_{[0,+\infty)^n}\Delta(\xi)^2 \prod_{j=1}^n \xi_j^\alpha e^{-n(t\xi_j + \xi^{-k})}d\xi_j.
\eeq
To evaluate $\widetilde{Z}_{n,k,\alpha}^{pLUE}$, we introduce monic orthogonal polynomials $\widetilde{p}_j$ satisfying
\beq\label{ortho p bar}
\int^\infty_0 \widetilde{p}_\ell(x) \widetilde{p}_m(x) \widetilde{w}(x) dx = \widetilde{h}_\ell \delta_{\ell m},\qquad \widetilde{w}(x) =x^\alpha\exp\left[-n \left(tx +x^{-k}\right)\right].
\eeq
Recalling the orthogonality conditions for $p_n$ given in \eqref{ortho p}, we easily obtain
\beq
\widetilde{p}_n(x) = t^{-n} p_n(tx), \qquad \widetilde{h}_n = t^{-(2n+\alpha+1)} h_n.
\eeq
A similar calculation for the Cauchy transform $q_n$ appearing in \eqref{Ysol} shows that
\beq
\widetilde{q}_n(z) = \frac{1}{2\pi i} \int^\infty_0 \frac{\widetilde{p}_n(x) \widetilde{w}(x)}{x-z} dx = t^{-\alpha -n} q_n(t z).
\eeq
Summarising, if we define
\beq
\widetilde{Y}(z) := \begin{pmatrix}\widetilde{p}_n(z)&\widetilde{q}_n(z)\\ -\frac{2\pi i}{\widetilde{h}_{n-1}} \widetilde{p}_{n-1}(z)& -\frac{2\pi i}{\widetilde{h}_{n-1}} \widetilde{q}_{n-1}(z)\end{pmatrix},
\eeq
we obtain the relationship
\beq
\widetilde{Y}(z) = t^{-n \sigma_3} t^{-\frac{\alpha}{2}\sigma_3} Y(tz) t^{\frac{\alpha}{2}\sigma_3}.
\eeq
Expressing $\widetilde{Z}_{n,k,\alpha}^{pLUE}$ as $\widetilde{Z}_{n,k,\alpha}^{pLUE} = \prod^{n-1}_{i=0} \widetilde{h}_i$
and taking the logarithmic derivative of both sides, we get
\begin{equation}
\label{eq:logdev1}
\frac{d}{dt} \log \widetilde{Z}_{n,k,\alpha}^{pLUE} = \int_0^\infty 
\left(\sum_{j=0}^{n-1} \frac{\widetilde{p}_j^2(x)}{\widetilde{h}_j}\right)\frac{\partial
  \widetilde{w}(x)}{\partial t} dx,  
\end{equation}
where we have made use of the orthogonality of the polynomials. This expression can be written as
\begin{equation}
\label{eq:kernint}  
\frac{d}{dt} \log \widetilde{Z}_{n,k,\alpha}^{pLUE} = -n \int_0^\infty x  \widetilde{K}_n(x,x)dx,
\end{equation}
where
\begin{equation}
  \label{eq:ker}
  \widetilde{K}_n(x,y) := \sqrt{\widetilde{w}(x)\widetilde{w}(y)}
  \sum_{j=0}^{n-1}\frac{\widetilde{p}_j(x)\widetilde{p}_j(y)}{\widetilde{h}_j}. 
\end{equation}
The expression for the kernel $\widetilde{K}_n$ in terms of $\widetilde{Y}$, similar to (\ref{KYeqn}), together
with the jump condition for $\widetilde{Y}$ implies the identity
\begin{equation}
  \widetilde{K}_n(x,x) = -\frac{1}{4\pi
    i}\left(\Tr\left(\widetilde{Y}^{-1}_{+}(x)\widetilde{Y}'_{+}(x)\sigma_3 \right) -
      \Tr\left(\widetilde{Y}^{-1}_{-}(x)\widetilde{Y}'_{-}(x) \sigma_3  \right)\right),
\end{equation}
for $x\in\mathbb R$,
which combined with the fact that $\Tr(\widetilde{Y}^{-1}(x)\widetilde{Y}'(x)\sigma_3) = t\Tr(Y^{-1}(t x)Y'(t x)\sigma_3)$, gives
\begin{equation}
\label{eq:kernint2}  
\frac{d}{dt} \log \widetilde{Z}_{n,k,\alpha}^{pLUE} = \frac{n}{4\pi i t} \int_0^\infty x  \left[\Tr(Y^{-1}_+(x)Y'_+(x)\sigma_3)-\Tr(Y^{-1}_-(x)Y'_-(x)\sigma_3) \right]dx.
\end{equation}

To prove the lemma, by (\ref{part sub}), we need to show that the above is equal to \[\frac{n}{2t} \Res_{z=\infty}\left(z \Tr \left(Y^{-1}(z)Y'(z)  \sigma_3 \right)\right).\]
To see this, consider the integral of $z \Tr \left(Y^{-1}(z)Y'(z)  \sigma_3 \right)$ over the contour $\mathcal C$ depicted in Figure \ref{fig:cont2}. One one hand, this integral is equal to zero because the integrand is analytic inside the integration contour.
\begin{figure}[t]
\centering
\includegraphics[width=3in]{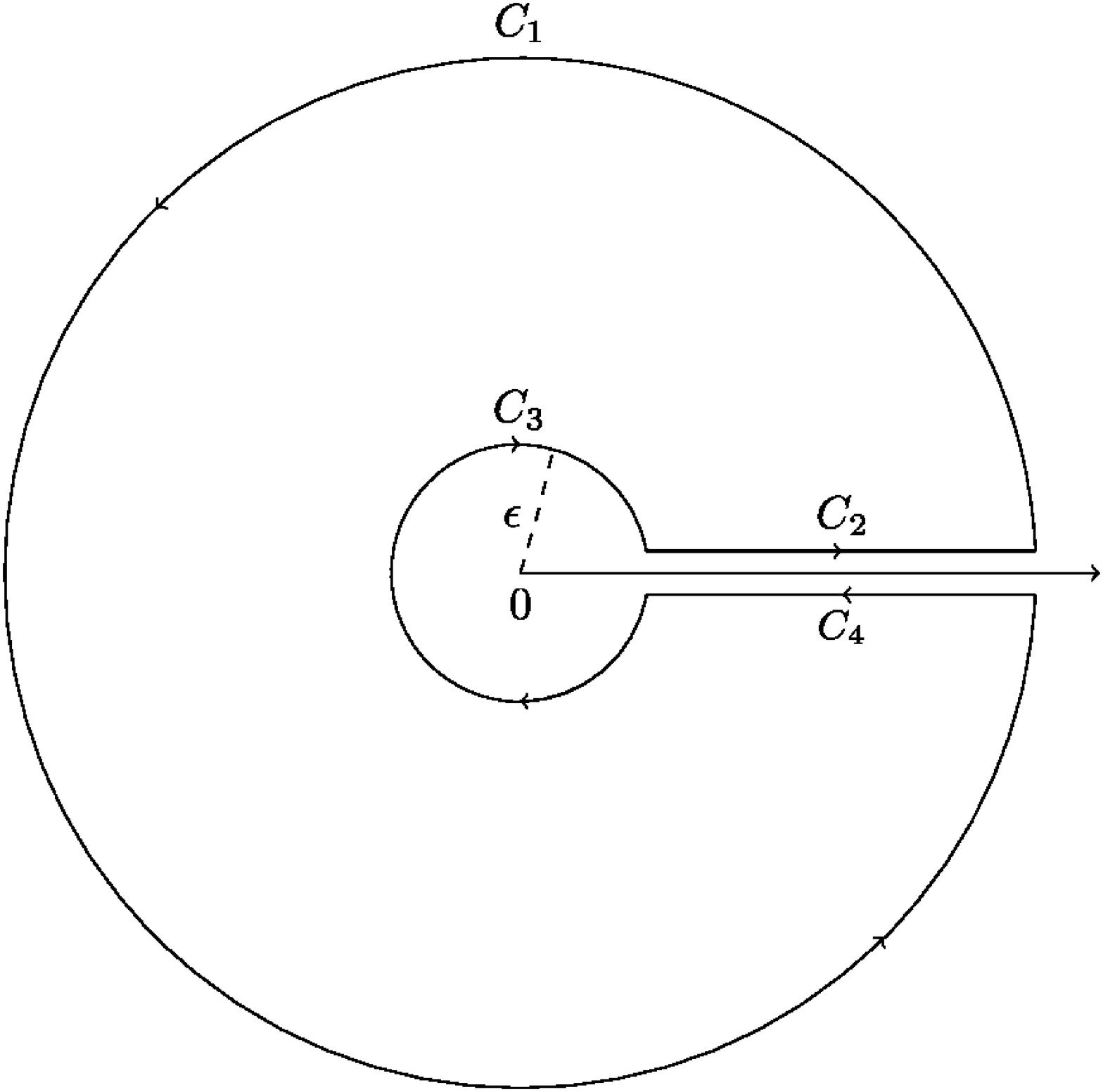}
\caption{The contour of integration $\mathcal{C}= C_1 + C_2 + C_3
  +C_4$}
\label{fig:cont2}
\end{figure}
On the other hand, condition (d) of the $Y$-RH problem implies that for sufficiently small $\epsilon$, $\Tr\left(Y^{-1}(z)Y'(z)\sigma_3\right)$ is bounded on the circle $C_3$, uniformly in $\epsilon$. Furthermore, condition (c) of the $Y$-RH problem implies
\[
\Tr\left(Y^{-1}(z)Y'(z)\sigma_3\right) = \bigO(z^{-2}), \quad z \to \infty.
\]
Therefore, letting $\epsilon\to 0$, letting the radius of the large circle $C_1$ tend to infinity, and letting $C_2, C_4$ approach the positive half-line, we arrive at the identity
\begin{multline}
-2\pi i  \Res_{z=\infty}\left(z \Tr \left(Y^{-1}(z)Y'(z)  \sigma_3 \right)\right)
\\+\int_0^{\infty}x \Tr\left(Y^{-1}_{+}(x)Y'_{+}(x)
\sigma_3\right)dx  - \int_0^{\infty}x \Tr\left(Y^{-1}_{-}(x)Y'_{-}(x)\sigma_3\right)dx=0.
\end{multline}
Substituting this into \eqref{eq:kernint2}, we obtain the result.
\end{proof}
We may now give the proof of the first part of Theorem \ref{theorem: partition}. We proceed by computing the residue
at infinity in \eqref{lem:dif_id}. The expression for the matrix $Y$ that is valid away from the support of the equilibrium measure (i.e.\ outside the disks near $0$ and $b$, and outside the lens-shaped region) is
\beq
Y(z) = e^{\frac{nl}{2}\sigma_3} R(z) P^{(\infty)}(z) e^{n(g(z) - \frac{l}{2})\sigma_3} e^{\frac{n}{2} \left(\frac{t}{z}\right)^k}.
\eeq
To obtain an expression for the residue, we compute the expansion as $z\rightarrow \infty$ of the above expression and substitute them into $\Tr(Y^{-1}Y'\sigma_3)$, giving as $z \rightarrow \infty$,
\begin{multline*}
\Tr(Y^{-1}Y'\sigma_3)(z) \\= \frac{1}{z}\left(-\frac{nkt^k}{z^k} + 2n(1+\frac{\mu_1}{z}) - \frac{1}{z}\Tr(P_1\sigma_3)- \frac{1}{z}\Tr(R_1\sigma_3) + \bigO(z^{-2})\right).
\end{multline*}
Here, the quantity $\mu_1$ arises from the expansion of the $g$-function and is the first moment of the equilibrium measure,
\beq
\mu_1 := \int^b_0 x \psi(x) dx,
\eeq
and
the quantities $P_1$ and $R_1$ are defined as the functions of $t$ appearing in the $z \rightarrow \infty$ asymptotics of $P^{(\infty)}$ and $R$,
\begin{align}
&P^{(\infty)}(z) = I + \frac{P_1}{z} + \bigO(z^{-2}), \\
&R(z) = I + \frac{R_1}{z} + \bigO(z^{-2}).
\end{align}
In the Laguerre case, we have $\mu_1 = 1$, $\Tr(P_1 \sigma_3) = -2\alpha$, and $R_1$ is given by \eqref{LUER}. Recall from Remark \ref{remark uniform0} that the asymptotics for $R_1$ are uniform also when $t\to 0$ very rapidly. Substituting all this in the differential identity (\ref{eq:dif_id}), and integrating between $0$ and $t$, we therefore obtain
\beq
\log \frac{Z_{n,k}^{pLUE}(t)}{Z_{n,k}^{pLUE}(0)} = n^2 t \delta_{k,1} + \frac{1}{2} \int^t_0 \left(\frac{1}{8}(1-4\alpha^2) - r(2^{-\frac{1}{k}}n^\frac{2k+1}{k} t')\right) \frac{dt'}{t'}+\bigO(n^{-1}),
\eeq
in the double scaling limit where $n\to\infty$, $t\to 0$ in such a way that $s_{t,n}=2^{-\frac{1}{k}}n^\frac{2k+1}{k} t\to s\in (0,+\infty)$.
After the change variables $\xi=2^{-\frac{1}{k}}n^\frac{2k+1}{k} t$, we obtain (\ref{as partition function pLUE}).

\subsection{The partition function in the pGUE}
\begin{proposition}\label{prop: GUE LUE}
Let $Z_{n,k,\alpha}^{pLUE}$ be the partition function in the pLUE, defined by \eqref{partition function pLUE}, and let  $Z_{n,k,\alpha}^{pGUE}$ be the partition function in the pGUE, defined by \eqref{partition function pGUE}. Then we have the following identities:
\begin{align}&\label{partition even}
\log Z_{2n,k,\alpha}^{pGUE}(t)=\log Z_{n,k,\alpha-\frac{1}{2}}^{pLUE}(t)+\log Z_{n,k,\alpha+\frac{1}{2}}^{pLUE}(t),\\
& \log Z_{2n+1,k,\alpha}^{pGUE}(t)=(n+\alpha+\frac{1}{2})\left[n\log\frac{2n}{2n+1} + (n+1)\log\frac{2n+2}{2n+1}\right]+ \nn\\
&\qquad\log Z_{n+1,k,\alpha-\frac{1}{2}}^{pLUE}\left(\left(\frac{2n+1}{2n+2}\right)^{\frac{k+1}{k}} t\right)+\log Z_{n,k,\alpha+\frac{1}{2}}^{pLUE}\left(\left(\frac{2n+1}{2n}\right)^{\frac{k+1}{k}} t\right).\label{partition odd}
\end{align}
\end{proposition}
\begin{proof}
Let $p_j(x;n,k,\alpha)$, $j=0,1,\ldots$, be the monic perturbed Laguerre polynomials defined by the orthogonality relations
\begin{equation}\label{ortho pLaguerre}
\int_0^{+\infty}p_j(x)p_\ell(x)x^\alpha e^{-n\left(x+\left(\frac{t}{x}\right)^k\right)}dx=h_j(n,\alpha,k)\delta_{j\ell}.
\end{equation}
Then it is well-known that
\begin{equation}\label{Z h LUE}
\log Z_{n,k,\alpha}^{pLUE}=\sum_{j=0}^{n-1}\log h_j(n,k,\alpha).
\end{equation}

A similar identity holds for the monic perturbed Hermite polynomials $\widehat p_j(n,k,\alpha)$, $j=0,1,\ldots$ defined by 
\begin{equation}\label{ortho pHermite}
\int_{\mathbb R}\widehat p_j(x)\widehat p_\ell(x)|x|^{2\alpha} e^{-\frac{n}{2}\left(x^2+\left(\frac{t}{x^2}\right)^k\right)}dx=\widehat h_j(n,k,\alpha)\delta_{j\ell}.
\end{equation}
We have 
\begin{equation}\label{Z h GUE}
\log Z_{n,k,\alpha}^{pGUE}= \sum_{j=0}^{n-1}\log \widehat h_j(n,k,\alpha).
\end{equation}

Using (\ref{ortho pLaguerre}) and (\ref{ortho pHermite}) and the uniqueness of the orthogonal polynomials defined by those orthogonality conditions, it is straightforward to derive the following identities (similar identities can be found in \cite[Appendix B]{ADDV} and \cite[Exercise 4.4]{Forrester}):
\begin{align}&\label{OP id}
\widehat p_{2j}(u;2n,k,\alpha)=p_{j}(u^2;n,k,\alpha-1/2),\\
&\label{OP id2}
\widehat p_{2j+1}(u;2n,k,\alpha)=up_{j}(u^2;n,k,\alpha+1/2).
\end{align}
This implies that
\begin{equation}\label{lc id}
\widehat h_{2j}(2n,k,\alpha)=h_j(n,k,\alpha-\frac{1}{2}),\qquad
\widehat h_{2j+1}(2n,k,\alpha)=h_j(n,k,\alpha+\frac{1}{2}).
\end{equation}
Using those relations together with (\ref{Z h GUE}) and (\ref{Z h LUE}), we obtain
\begin{eqnarray}
\log Z_{2n,k,\alpha}^{pGUE}&=&\sum_{j=0}^{n-1}\log \widehat h_{2j}(2n,k,\alpha)+\sum_{j=0}^{n-1}\log \widehat h_{2j+1}(2n,k,\alpha)\\
&=&\sum_{j=0}^{n-1}\log h_j(n,k,\alpha-\frac{1}{2})+\sum_{j=0}^{n-1}\log h_j(n,k,\alpha+\frac{1}{2})\\
&=& \log Z_{n,k,\alpha-\frac{1}{2}}^{pLUE}+\log Z_{n,k,\alpha+\frac{1}{2}}^{pLUE}.
\end{eqnarray}
Similarly,
\begin{align}
&\log Z_{2n+1,k,\alpha}^{pGUE}=\sum_{j=0}^{n}\log \widehat h_{2j}(2n+1,k,\alpha)+\sum_{j=0}^{n-1}\log \widehat h_{2j+1}(2n+1,k,\alpha)\nonumber\\
&\quad =\sum_{j=0}^{n}\log h_j(n+\frac{1}{2},k,\alpha-\frac{1}{2})+\sum_{j=0}^{n-1}\log h_j(n+\frac{1}{2},k,\alpha+\frac{1}{2})\nonumber\\
&\quad =\log A(n+1,n+\frac{1}{2},k,\alpha-\frac{1}{2})+\log A(n,n+\frac{1}{2},k,\alpha+\frac{1}{2})\label{Z odd},
\end{align}
where 
\begin{equation}
A(n,N,k,\alpha)=\frac{1}{n!}\int_{[0,+\infty)^n}\Delta(x)^2 \prod_{j=1}^n x_j^\alpha e^{-N\left(x_j+\left(\frac{t}{x_j}\right)^k\right)}dx_j.
\end{equation}
We have
\begin{equation}
A(n,N,k,\alpha)= \left(\frac{n}{N}\right)^{n^2+n\alpha}Z_{n,k,\alpha}^{pLUE} \left(\left(\frac{N}{n}\right)^{\frac{k+1}{k}} t\right).
\end{equation}
Substituting this into (\ref{Z odd}) completes the proof.
\end{proof}
The proof of \eqref{as partition function pGUE} can now easily be completed by substituting (\ref{as partition function pLUE}) into \eqref{partition even}-\eqref{partition odd}.

\section*{Acknowledgements} 
MA and TC were supported by the European Research Council under the European Union's Seventh Framework Programme (FP/2007/2013)/ ERC Grant Agreement n.\, 307074, by the Belgian Interuniversity Attraction Pole P07/18, and by F.R.S.-F.N.R.S. FM was partially supported by EPSRC research grant no. EP/L010305/1.


\begin{thebibliography}{99}
\bibitem{ADDV} G. Akemann, D. Dalmazi, P.H. Damgaard, and J.J.M.
    Verbaarschot, QCD3 and the Replica Method, {\em Nucl.Phys. B} {\bf 601} (2001),
    77--124.
\bibitem{AVV} G. Akemann, D. Villamaina, and P. Vivo, A singular-potential random matrix model arising in mean-field glassy systems, {\em{Phys. Rev. E}} {\bf{89}}, 062146 (2014).
\bibitem{BEH06}M. Bertola, B. Eynard, and J. Harnad, Semiclassical orthogonal polynomials, matrix models, and isomonodromic tau functions, {\em{Comm. Math. Phys.}} {\bf{263}} (2006), no. 2, 401–-437.
\bibitem{BMM} L. Brightmore, F. Mezzadri, and M.Y. Mo, A matrix model
  with a singular weight and Painlev\'e III, {\em Commun. Math. Phys.}
  to appear. arxiv:1003.2964.
  \bibitem{CI} Y. Chen and A. Its, Painlev\'e III and a singular linear statistics in Hermitian random matrix ensembles I, {\em J. Approx. Theory} {\bf 162} (2010), no. 2, 270–-297.
  \bibitem{BFB1} P. W. Brouwer, K. M. Frahm, and C. W. Beenakker, Quantum mechanical time-delay matrix in chaotic scattering, {\em{Phys. Rev. Lett.}} {\bf{78}} (1997), 25, 4737.
\bibitem{BFB2} P. W. Brouwer, K. M. Frahm, and C. W. Beenakker, Distribution of the quantum mechanical time-delay matrix for a chaotic cavity, {\em{Waves Random Media}} {\bf{9}} (1999), 91-–104.
\bibitem{DKM}
    P. Deift, T. Kriecherbauer, and K.T-R McLaughlin,
    New results on the equilibrium measure for logarithmic potentials
    in the presence of an external field,
    {\em J. Approx. Theory} {\bf 95} (1998), 388--475.
\bibitem{DKMVZ2}
    P. Deift, T. Kriecherbauer, K.T-R McLaughlin, S. Venakides, and X. Zhou,
    Uniform asymptotics for polynomials orthogonal with respect to
    varying exponential weights and applications to universality
    questions in random matrix theory,
    {\em Comm. Pure Appl. Math.} {\bf 52} (1999), 1335--1425.
\bibitem{DKMVZ1}
    P. Deift, T. Kriecherbauer, K.T-R McLaughlin, S. Venakides,
    and X. Zhou,
    Strong asymptotics of orthogonal polynomials with respect to
    exponential weights,
    {\em Comm. Pure Appl. Math.} {\bf 52} (1999), 1491--1552.
\bibitem{DZ1}
        P. Deift and X. Zhou,
        A steepest descent method for oscillatory Riemann-Hilbert problems.
            Asymptotics for the MKdV equation,
        {\em Ann. Math.} {\bf 137} (1993), no. 2, 295--368.
\bibitem{Di Francesco:1993nw}
  P.~Di Francesco, P.~H.~Ginsparg and J.~Zinn-Justin,
  2-D Gravity and random matrices,
  {\em Phys.\ Rept.\ } {\bf 254} (1995), no.\ 1-2.
\bibitem{FIKN} A.S. Fokas, A.R. Its, A.A. Kapaev, and V.Yu.
Novokshenov, ``\,Painlev\'e transcendents: the Riemann-Hilbert
approach", AMS Mathematical Surveys and Monographs \textbf{128}
(2006).
\bibitem{FokasItsKitaev}
    A.S. Fokas, A.R. Its, and A.V. Kitaev,
    The isomonodromy approach to matrix models in 2D quantum gravity,
    {\em Comm. Math. Phys.} {\bf 147} (1992), 395--430.
  \bibitem{FokasMuganZhou}
        A.S. Fokas, U. Mugan, and X. Zhou,
        On the solvability of Painlev\'e I, III and V,
        {\em Inverse Problems} {\bf 8}, no. 5, (1992), 757--785.
    \bibitem{FokasZhou}
        A.S. Fokas and X. Zhou,
        On the solvability of Painlev\'e II and IV,
        {\em Comm. Math. Phys.} {\bf 144}, no. 3, (1992), 601--622.  
        \bibitem{Forrester} P.J. Forrester, ``\,Log-gases and random
    matrices",  London Mathematical Society Monographs Series {\bf 34}, Princeton University Press, Princeton, NJ (2010).
        \bibitem{GT} A. Grabsch and C. Texier, Capacitance and charge relaxation resistance of chaotic cavities - Joint distribution of two linear statistics in the Laguerre ensemble of random matrices, arxiv:1407.3302.
 \bibitem{IKO} A. Its, A. Kuijlaars, and J. Ostensson, Critical edge behavior in unitary random matrix ensembles and the thirty fourth Painlev\'e transcendent, {\em Internat. Math. Res. Notices} {\bf 2008} (2008), article ID rnn017, 67 pages.   
    \bibitem{Klebanov:2003wg}
  I.~R.~Klebanov, J.~M.~Maldacena, and N.~Seiberg,
  Unitary and complex matrix models as 1-d type 0 strings,
 {\em Commun.\ Math.\ Phys.\ } {\bf 252} (2004), 275–-323.
  \bibitem{OsiKan}  V. A. Osipov and E. Kanzieper, Are bosonic replicas faulty? {\em Phys Lett. Rev.} {\bf 99} (2007), 
050602.    
    
    
    
\bibitem{SaffTotik}
    E.B. Saff and V. Totik,
    ``\,Logarithmic Potentials with External Fields",
    Springer-Verlag, New-York (1997).
\bibitem{Sakka} A.H. Sakka, Linear problems and hierarchies of Painlev\'e equations, {\em J. Phys. A: Math. Theor.} {\bf 42} (2009), 025210. 
\bibitem{Seiberg:2004ei}
  N.~Seiberg and D.~Shih,
  Flux vacua and branes of the minimal superstring,
{\em JHEP} {\bf 0501} (2005), 055.
\bibitem{MT} C. Texier and S.N. Majumdar, Wigner time-delay distribution in chaotic cavities and freezing transition, {\em{Phys. Rev. Lett.}}, {\bf{110}} (2013), 250602.
\bibitem{Vanlessen2}
    M. Vanlessen, Strong asymptotics of Laguerre-type orthogonal polynomials and applications
    in random matrix theory, {\em Constr.
    Approx.} {\bf 25} (2007), no. 2, 125–-175.
\bibitem{XDZ} S.-X. Xu, D. Dai, and Y.-Q. Zhao, Critical edge behavior and the Bessel to Airy transition in the singularly perturbed Laguerre unitary ensemble, arxiv:1309.4354.
\bibitem{XDZ2} S.-X. Xu, D. Dai, and Y.-Q. Zhao, Painlev\'e III asymptotics of Hankel determinants for a singularly perturbed Laguerre weight, arxiv:1407.7334.



 






\end{thebibliography}
\end{document}